\documentclass[12pt, draftclsnofoot, onecolumn]{IEEEtran}
\usepackage{epsfig,latexsym}
\usepackage{float}
\usepackage{indentfirst}
\usepackage{amsmath}
\usepackage{bm}
\usepackage{amssymb}
\usepackage{times}
\usepackage{graphicx}

\usepackage{algorithm}
\usepackage[noend]{algpseudocode}
\usepackage{subfigure}
\usepackage{psfrag}
\usepackage{hyperref}
\usepackage{cite}
\usepackage{lastpage}
\usepackage{fancyhdr}
\usepackage{color}
 \usepackage{amsthm}
\usepackage{bigints}
\newtheorem{Proposition}{Proposition}
\newtheorem{Lemma}{Lemma}
\newtheorem{lemma}[Lemma]{$\mathbf{Lemma}$}
\newtheorem{proposition}[Proposition]{Proposition}

\newcounter{problem}
\newcounter{save@equation}
\newcounter{save@problem}
\makeatletter

\begin{document}
\title{\vspace{-0.5em} \huge{Impact of NOMA on Age of Information: \\A Grant-Free Transmission Perspective  }}

\author{ Zhiguo Ding, \IEEEmembership{Fellow, IEEE},   Robert Schober, \IEEEmembership{Fellow, IEEE}, and H. Vincent Poor, \IEEEmembership{Life Fellow, IEEE}    \thanks{ 
  
\vspace{-2em}

    Z. Ding and H. V. Poor are  with the Department of
Electrical and Computer Engineering, Princeton University, Princeton, NJ 08544,
USA. Z. Ding
 is also  with the School of
Electrical and Electronic Engineering, the University of Manchester, Manchester, UK (email: \href{mailto:zhiguo.ding@manchester.ac.uk}{zhiguo.ding@manchester.ac.uk}, \href{mailto:poor@princeton.edu}{poor@princeton.edu}).
  R. Schober is with the Institute for Digital Communications,
Friedrich-Alexander-University Erlangen-Nurnberg (FAU), Germany (email: \href{mailto:robert.schober@fau.de}{robert.schober@fau.de}).
 

  }\vspace{-2em}}
 \maketitle

\vspace{-1em}
\begin{abstract}
The aim of this paper is to  characterize the impact of non-orthogonal multiple access (NOMA) on the age of information (AoI) of grant-free transmission. In particular, a low-complexity form of NOMA, termed NOMA-assisted random access, is applied to grant-free transmission in order to   illustrate the two benefits of   NOMA for   AoI reduction, namely increasing    channel access and reducing user collisions. Closed-form analytical expressions  for the AoI achieved by NOMA assisted grant-free transmission are obtained, and asymptotic studies are carried out to demonstrate that the use of the simplest form of NOMA is already sufficient to   reduce   the AoI of orthogonal multiple access (OMA) by more than $40\%$. In addition, the developed analytical expressions  are also shown to be useful  for optimizing the users' transmission attempt probabilities, which are   key parameters for grant-free transmission.  
\end{abstract}\vspace{-1em}

 \section{Introduction}
 Grant-free transmission is an important  enabling technique to support the   sixth-generation (6G) services, including ultra massive machine type communications (umMTC) and enhanced ultra-reliable low
latency communications (euRLLC) \cite{9031550,you6g}. Unlike conventional  grant-based transmission, grant-free transmission enables users to avoid     multi-step signalling and directly transmit data signals together with access control information, which can reduce the system overhead significantly, particularly for    scenarios with massive users requiring short-package transmission. Grant-free transmission can be realized by applying the random access protocols developed for conventional computer networks, such as ALOHA random access \cite{7891878,9537931}. Alternatively, massive multi-input multi-output (MIMO) can also be applied to support grant-free transmission by using the excessive  spatial degrees of freedom offered by massive MIMO \cite{8323218,8454392,8444464}.  
 
 Recently, the application of non-orthogonal multiple access (NOMA) to grant-free transmission has received significant  attention due to the following two reasons. First,  the NOMA principle is highly compatible, and     the use of NOMA can significantly  improve the reliability and   spectral efficiency of   random access and massive MIMO based grant-free protocols    \cite{8533378, 9496190, 8886595,9022993}. Second, more importantly,   the use of NOMA alone is sufficient to  support grant-free transmission. For example, NOMA-based grant-free transmission   has been proposed in \cite{9707731}, where a Bayesian learning based   scheme has been designed to ensure successful multi-user detection, even if the number of active grant-free users is unknown.   The principle of NOMA can also be used to develop   so-called semi-grant-free transmission protocols, where the bandwidth resources which would be solely occupied by grant-based users are released for  supporting  multiple grant-free users in a distributed manner \cite{8662677}.  In addition, NOMA-based grant-free transmission has also been shown to be robust and efficient in various communication scenarios, such as satellite  communication networks, secure Internet of Things (IoT), intelligent reflecting surface (IRS) networks, marine communication systems, etc., see \cite{9920182,9837307, 9686696,9122624}. 

The aim of this paper is to characterize the impact of NOMA on the performance of grant-free transmission with respect to a recently developed new performance metric, termed the age of information (AoI) \cite{8000687,6195689,9606181}. In particular, the AoI describes  the freshness of  data updates collected in the network, and   is an important metric to measure the success of the   6G services, including  umMTC and euRLLC.  We note that most existing works have focused on the impact of NOMA on grant-based networks \cite{crnomaaoi, 9130084,9771565, 9840754,9734735}.   For example,     for two-user grant-based  networks, the capability of  NOMA to reduce the AoI has been shown to be related to the spectral efficiency gain of NOMA over orthogonal multiple access (OMA) \cite{8845254}. To  the authors' best knowledge, there is only a single  existing work which applied  NOMA to reduce the AoI of grant-free transmission  \cite{ 9442815}, where     the strong assumption that the base station  estimates all   users' channel state information (CSI) was made.     

In this paper, the impact of NOMA on the AoI of grant-free transmission is investigated from the perspective of performance analysis, which is different from the existing work focusing on resource allocation   \cite{9508961}. In particular, a low-complexity form of NOMA,  which was originally  termed NOMA-assisted random access   \cite{jsacnoma10} and recently also termed  ALOHA with successive interference cancellation (SIC) \cite{8371304, 83713042}, is adopted in order to   illustrate the two   benefits of   NOMA for   AoI reduction, namely increasing    channel access and reducing user collisions.   The key element of the proposed performance analysis is the modelling of         the   channel competition among the grant-free users   as a Markov chain, which is different from the performance analysis   for   grant-based NOMA networks    \cite{8845254,crnomaaoi,9130084,9771565}. 
The main contributions of this paper are two-fold:
\begin{itemize}
\item Analytical expressions  for the AoI achieved by   NOMA assisted grant-free transmission are obtained, by rigorously characterizing the state transition probabilities of the considered Markov chain.    We note that  by using NOMA-assisted random access, the base station   creates multiple preconfigured receive signal-to-noise ratio (SNR) levels, which makes NOMA-assisted grant-free transmission similar to multi-channel ALOHA.   As a result,    the calculation of the state transition probabilities for the NOMA case is more challenging  than that for the OMA case, which can be viewed as single-channel ALOHA.   By exploiting  the properties of the considered Markov chain and also the characteristics   of SIC,   closed-form expressions for   the state transition probabilities are developed for NOMA assisted grant-free transmission. 

\item Valuable insights regarding   the relative performance    of NOMA and OMA assisted grant-free transmission are also obtained. For example, for the case where   users always  have updates to deliver,  asymptotic  expressions  are developed to    demonstrate that the use of NOMA can almost halve the AoI achieved by OMA, even if  the simplest form of   NOMA   is implemented. In addition, the optimal choices of the users' transmission probabilities for random access with  NOMA and OMA are  obtained and compared. This study reveals   that     NOMA-assisted grant-free transmission is fundamentally different from  multi-channel ALOHA  due to the use of SIC.  Furthermore,  simulation   results are provided to verify the developed analytical expressions  and also demonstrate  that the use of NOMA can significantly reduce the AoI compared to  OMA, particularly for the case of low transmit SNR and a massive number of grant-free users. 
\end{itemize}



\section{ NOMA Assisted Grant-Free Transmission}\label{section 2}
Consider a grant-free communication network with $M$   users,    denoted by ${\rm U}_m$, $1\leq m\leq M$, communicating  with the same base station.  Assume that  each time frame comprises  $N$ time slots, each of duration of    $T$ seconds,   where the $n$-th time slot of the $i$-th frame is denoted by ${\rm TS}_i^n$, and the  starting time of ${\rm TS}_i^n$ is denoted by   $t_i^n$, $1\leq n\leq N$ and $i\geq 1$, as shown   in Fig. \ref{fig1}.

    \begin{figure}[t]\centering \vspace{-2em}
    \epsfig{file=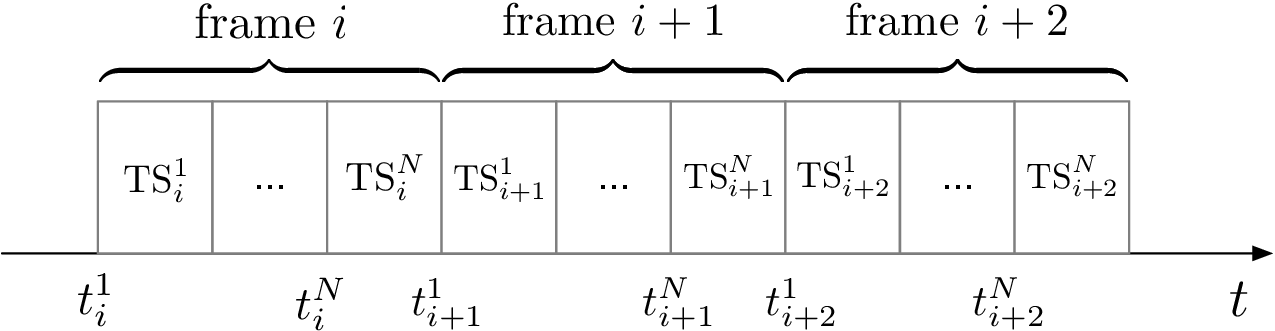, width=0.6\textwidth, clip=}\vspace{-0.5em}
\caption{ Considered slotted time frame structure. 
  \vspace{-1em}    }\label{fig1}   \vspace{-1.5em} 
\end{figure}

\subsection{Data Generation Models}
For the considered grant-free transmission scenario, each user tries  to deliver one update to the base station in each time frame. When the users' updates are generated depends on which of the following two data generation models is used \cite{crnomaaoi}.
\subsubsection{Generate-at-request (GAR)} For GAR, the base station requests all users to simultaneously generate their updates  at the beginning of each time frame. GAR is applicable  to     many important   IoT applications, such as structural  health monitoring and autonomous driving.  

\subsubsection{Generate-at-will (GAW)}  For GAW, a user's update is generated right before its transmit time slot. GAW has been commonly used  in the AoI literature, since it ensures  that  the delivered updates are freshly generated. 

In this paper, GAR will be focused on due to the following two reasons. First, the AoI expression for  grant-free transmission for GAW can be straightforwardly obtained from that for GAR, as shown in the next section.   Second,  if there are retransmissions     within one time frame, GAW requires   a user to repeatedly generate updates,   and hence, causes a higher  energy consumption than GAR. For grant-based transmission, this increase in energy consumption   is not severe since the number of retransmissions in one frame is small \cite{crnomaaoi}. However, in the grant-free case, a user might have to carry out a large number of retransmissions due to potential collisions, which means that   GAW can cause a significantly  higher  energy consumption than  GAR.  


\subsection{Channel Access Modes} \label{subsectionx1}
\subsubsection{Orthogonal Multiple Access (OMA)}
 Based on OMA, a user's transmission can be successful only if it solely occupies the bandwidth resource block, i.e., a time slot.  A simple example of OMA based grant-free transmission is   slotted ALOHA, as described in the following\footnote{ In this paper, a simple slotted ALOHA scheme is considered, where each user can use any of the time slots in the frame. For  more sophisticated random access schemes, e.g., irregular repetition slotted ALOHA (IRSA) \cite{9358219,9376691,9377549}, which allows a user to choose  a subset of time slots for transmission, the principle of NOMA can be also  applied as an add-on, e.g., a group of users, instead of a single user, share the same subset of time slots. }. 

Prior to ${\rm TS}_i^n$, assume that   $j$ users   have successfully delivered their updates to the base station.   For grant-free transmission, each of the remaining $M-j$ users will independently make a transmission attempt with the same transmit power, denoted by $P$\footnote{Because the noise power is assumed to be normalized to one, $P$ can also be   interpreted as the transmit SNR.   }, and the same transmission probability, denoted by $\mathbb{P}_{\rm TX}$.   $\mathbb{P}_{\rm TX}$ can be      based on a fixed choice, or be  state-dependent, i.e.,   $\mathbb{P}_j=\frac{1}{M-j}$ \cite{9695972}.  It is assumed that the base station can inform the users about the outcome of their updates via a dedicated control channel.     

There are three possible events which cause an update failure for a user: i)     the user does not make a transmission  attempt; ii)    a collision occurs, i.e.,   there are more than one concurrent transmissions; iii) an outage occurs due to  the user's weak channel condition, i.e., $\log(1+P|h_{m}^{i,n}|^2)\leq R$, where $R$ denotes the user's target data rate, and  $h_{m}^{i,n}$ denotes ${\rm U}_m$'s channel gain in ${\rm TS}_i^n$. In this paper, all users are assumed to have the identical  target data rates, and their channel gains are assumed to be independent   and identically   complex Gaussian distributed with zero mean and unit variance.

\subsubsection{Non-orthogonal Multiple Access (NOMA)}
 With NOMA, a user  can   still succeed in its updating, even if multiple users choose the same time slot. Recall that the principle of NOMA can be implemented in different forms. For  the  purpose of illustration, a particular form of NOMA, termed NOMA-assisted random access, is adopted  to reduce the AoI of grant-free transmission \cite{jsacnoma10}.   In particular, prior to transmission, the base station configures $K$ receive SNR levels, denoted by   $P_1\geq \cdots\geq P_K$.  If ${\rm U}_m$ chooses $P_k$ during ${\rm TS}_i^n$, it will scale  its  transmitted signal by $\sqrt{\frac{P_k}{|h_m^{i,n}|^2}}$.\footnote{ One benefit of this form of NOMA   is that the base station does not need to estimate all   users' CSI for implementing SIC, an assumption commonly required   in, e.g., \cite{9508961,9442815}.  For the adopted form of NOMA, the users are assumed to have access to their own CSI only.   In practice, this CSI assumption can be realized by asking the base station to broadcast pilot signals at the beginning of a time slot, where the users can perform channel estimation individually.  }  The base station carries out SIC by decoding the signal delivered  at   SNR level, $P_k$, before  decoding the one at $P_{k+1}$, $1\leq k \leq K-1$. The SNR levels are preconfigured to guarantee the success of SIC, i.e.,   the following conditions need to be satisfied:
\begin{align}\label{cr rate}
\log\left(1+\frac{P_k}{1+(M-1)P_{k+1}} \right)=R, \quad 1\leq k \leq K-1,
\end{align}
and $\log\left(1+P_K\right)=R$, 
which means  $P_K=2^R-1$ and $P_k=\left(2^R-1\right)\left(1+(M-1)P_{k+1}\right)$, where   the noise power is assumed to be normalized to one.   We note that the condition in \eqref{cr rate} is   stricter than the condition      $\log\left(1+\frac{P_k}{1+\sum_{i=k+1}^{K} P_i} \right)=R$, and   ensures the success of SIC, even if one user chooses $P_k$ and the remaining    users choose the SNR level which contributes the most interference, i.e., $P_{k+1}$. We further note that  the case in which  all $M-1$ remaining users choose $P_{k+1}$ is the worst case, since some users   may   choose     SNR levels other than $P_{k+1}$ or even decide not to transmit at all.

Again assume that there are $j$ users which have successfully sent their updates to the base station prior to ${\rm TS}_i^n$. Each of the remaining $M-j$ users will first randomly choose an SNR level with equal probability, denoted by $\mathbb{P}_k=\frac{1}{K}$, and    independently make a transmission attempt with       probability $\mathbb{P}_{\rm TX}$. For illustrative purposes,  assume that ${\rm U}_m$ is among the $M-j$ remaining  users, and chooses  $P_k$.  The possible events which cause ${\rm U}_m$'s update to fail   are listed as follows:
\begin{itemize}
\item The user does not make an attempt for transmission;
\item The receive SNR level chosen by the user is not feasible due to the user's transmit power budget, i.e., $P_k$ is not feasible for ${\rm U}_m$ in ${\rm TS}_i^n$ if $\frac{P_k}{|h_{m}^{i,n}|^2}> P$;
\item Another user   also chooses $P_k$, which leads to a collision at $P_k$ and hence a failure at the $k$-th stage of SIC;
\item Prior to the $k$-th stage of SIC, SIC has already been terminated  due to one or more failures in the previous SIC stages. 
\end{itemize}
 
We note that for both the OMA and NOMA cases,     a user keeps re-sending its update to the base station until either the user succeeds  or the time frame is finished. 

\subsection{AoI Model}
AoI is an important performance metric for quantifying  the freshness of the updates delivered to the base station.  We note that for the considered grant-free scenario, all the users experience the same AoI. Therefore, without loss of generality,   ${\rm U}_1$'s instantaneous AoI at time $t$ is focused on and   defined as follows \cite{8000687}: 
\begin{align}
\Delta(t) = t - T(t),
\end{align}
where  $T(t)$ denotes   the
generation time of ${\rm U}_1$'s  freshest update   successfully delivered  to 
the base station.   ${\rm U}_1$'s      average AoI of the considered network is given by
\begin{align}
\bar{\Delta} =  \underset{T_\Delta \rightarrow\infty}{\lim} \frac{1}{T_\Delta}\int^{T_\Delta}_{0}\Delta(t) dt.
\end{align} 

The AoI achieved by OMA and NOMA assisted grant-free transmission will be analyzed in the following section.

    \begin{figure}[t]\centering \vspace{-2em}
    \epsfig{file=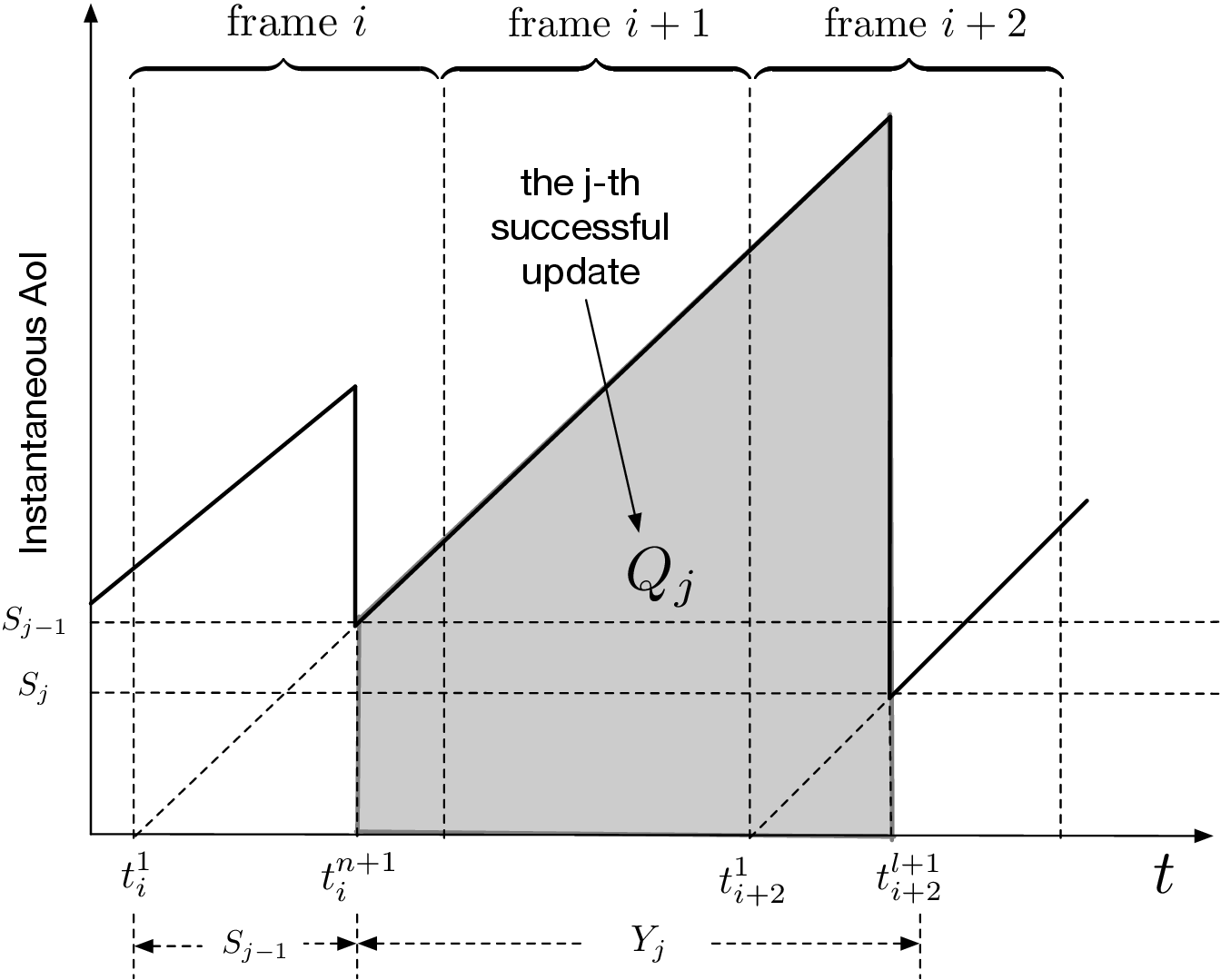, width=0.5\textwidth, clip=}\vspace{-0.5em}
\caption{  Illustration of AoI evolution with grant-free   transmission for GAR.  
  \vspace{-1em}    }\label{fig2}   \vspace{-1.5em} 
\end{figure}

\section{AoI of Grant-Free Transmission}\label{section 3}
As discussed in the previous section, ${\rm U}_1$ is treated as the tagged user and its AoI will be focused on in this section, without loss of generality. For the example shown in Fig. \ref{fig2},    ${\rm U}_1$  successfully sends its updates  to the base station in   ${\rm TS}_n^{i}$ of frame $i$ and ${\rm TS}_l^{i+2}$ of frame $i+2$, but   fails  in frame $i+1$. For the AoI analysis,  the following metrics are required:  
\begin{itemize}
\item $S_{j}$: The time duration between the generation time and the receive  time of the $j$-th successful update. For the example shown in Fig. 2,   $S_{j-1}=nT$ and $S_j=lT$.
\item $Y_j$:  The time duration between the $(j-1)$-th and the $j$-th successful updates.   For the example shown in Fig. 2,   $Y_j=(N-n)T+NT+lT$.
\item $X_j$: The number of frames between the $(j-1)$-th and the $j$-th successful updates. An example of $X_j=2$ is shown in Fig. 2.
\end{itemize}
We note that in the literature of random access, $S_j$ is termed the service delay and $Y_j$ is termed the inter-departure time \cite{9695972}.

By using the aforementioned metrics, for GAR,   ${\rm U}_1$'s  averaged AoI is given by
\begin{align} 
\bar{\Delta} =& \underset{J\rightarrow \infty}{\lim}\frac{\sum^{J}_{j=1} Q_j}{\sum^{J}_{j=1} Y_j}= \underset{J\rightarrow \infty}{\lim}\frac{\sum^{J}_{j=1} \left(S_{j-1}Y_j+\frac{1}{2}Y_j^2 \right)}{\sum^{J}_{j=1} Y_j} = \label{aoi 1}
  \frac{\mathcal{E}\{S_{j-1}Y_j\}}{\mathcal{E}\{Y_j\}}+\frac{\mathcal{E}\{Y_j^2\}}{2\mathcal{E}\{Y_j\}},
\end{align}
where $J$ denotes the total number of successful updates,  $Q_j$ denotes the area of the shaded region shown in Fig.~\ref{fig2}, 
$\mathcal{E}\{Y_j^2\}= \underset{J\rightarrow \infty}{\lim}\frac{\sum^{J}_{j=1} Y_j^2}{J}$,   $\mathcal{E}\{Y_j\}= \underset{J\rightarrow \infty}{\lim}\frac{\sum^{J}_{j=1} Y_j}{J}$, and $\mathcal{E}\{\cdot\}$ denotes the expectation.

{\it Remark 1:} We note that the AoI expressions in \eqref{aoi 1} of this paper and  \cite[Eq. (3)]{9695972} are consistent. The reason why there is an extra factor of $\frac{1}{2}$   in \cite[Eq. (3)]{9695972}  is that the users' instantaneous AoI was assumed to be discrete-valued  in \cite{9695972}, instead of continuous-valued  as   in this paper.   

{\it Remark 2:} For GAW,  the user's averaged AoI is given by
\begin{align}
\bar{\Delta}^{\rm GAW} =&   \underset{J\rightarrow \infty}{\lim}\frac{\sum^{J}_{j=1} \left(TY_j+\frac{1}{2}Y_j^2\right)}{\sum^{J}_{j=1} Y_j}= T + \frac{\mathcal{E}\{Y_j^2\}}{2\mathcal{E}\{Y_j\}},
\end{align} which is simpler than    the AoI expression in \eqref{aoi 1}. Therefore, the analytical results developed for GAR are straightforwardly applicable to the case for GAW.

\subsection{Generic Expressions for $\mathcal{E}\{S_{j-1}Y_j\}$,  $ \mathcal{E}\{Y_j\} $, and $ \mathcal{E}\{Y_j^2\}$}
As shown in \eqref{aoi 1}, the AoI is a function of $\mathcal{E}\{S_{j-1}Y_j\}$,    $ \mathcal{E}\{Y_j\} $, and $ \mathcal{E}\{Y_j^2\}$, and     generic expressions for these metrics   have been derived in \cite{9695972}, and will be briefly introduced  in this subsection. 
In particular, the considered grant-free transmission can be modelled by a Markov chain with  $M+1$ states, denoted by $s_k$, $0\leq k \leq M$. In particular,   $s_k$, $0\leq k \leq M-1$, denotes the transient state, where   $k$   users, other than  ${\rm U}_1$, have successfully delivered their updates to the base station.   $s_M$ means that ${\rm U}_1$ has successfully delivered its update to the base station. 
Define  the state transition  probability from $s_j$ to $s_i$ by  $P_{j,i}$, $0\leq i,j\leq M$. Build an $M\times M$ matrix, denoted by $\mathbf{P}$, whose element in the  $(i+1)$-th column and the $(j+1)$-th row is $P_{j,i}$, $0\leq i,j\leq M-1$. Furthermore, build an $M\times 1$ vector, denoted by $\mathbf{p}$, whose $(j+1)$-th element is  $  P_{j,M}$. 
Once $\mathbf{P}$ and $\mathbf{p}$ are available, $\mathcal{E}\{S_{j-1}Y_j\}$,  $ \mathcal{E}\{Y_j\} $, and $ \mathcal{E}\{Y_j^2\}$ can be obtained as follows. 

Denote by $Z$   the number of time slots  required by ${\rm U}_1$   to successfully deliver its update to its base station. Then, the     probability mass function (pmf) of $Z$   is given by
\begin{align}
  \mathbb{P}\left( Z=n\right) = \mathbf{s}_0^T\mathbf{P}_M^{n-1} \mathbf{p},\quad n=1, 2, \cdots,
\end{align}
where $\mathbf{s}_0=\begin{bmatrix}1 & \mathbf{0}_{1\times (M-1)} \end{bmatrix}^T$ denotes the initial probability vector and $\mathbf{0}_{m\times n}$ denotes an all-zero $m\times n$ matrix. 
Therefore, the probability that ${\rm U}_1$ cannot complete an update within one frame is   given by $
P_{\rm fail} = \mathbb{P}\left( Z>N\right)  = \mathbf{s}_0^T\mathbf{P}_M^N\mathbf{1} $, where $\mathbf{1}$ denotes an $M\times 1$ all-one vector. 
Therefore, the pmf of    access delay, $S_j$,  can be written as follows:
\begin{align}
\mathbb{P}(S_j=nT) = \frac{\mathbb{P}\left( Z=n\right)}{1-P_{\rm fail}} = \frac{\mathbf{s}_0^T\mathbf{P}_M^{n-1} \mathbf{p}}{1- \mathbf{s}_0^T\mathbf{P}_M^N\mathbf{1} }, 
\end{align}
for $1\leq n \leq N$, 
which means $
\mathcal{E}\{S_j\} = T\sum^{N}_{n=1}n  \frac{\mathbf{s}_0^T\mathbf{P}_M^{n-1} \mathbf{p}}{1- \mathbf{s}_0^T\mathbf{P}_M^N\mathbf{1} }$
and
$
\mathcal{E}\{S_j^2\} =T^2 \sum^{N}_{n=1}n^2  \frac{\mathbf{s}_0^T\mathbf{P}_M^{n-1} \mathbf{p}}{1- \mathbf{s}_0^T\mathbf{P}_M^N\mathbf{1} }
$. 

Similarly, the pmf of $X_j$ is given by 
\begin{align}
\mathbb{P}(X_j=n)=P_{\rm fail}^{n-1}(1-P_{\rm fail}),
\end{align}
which means that $\mathcal{E}\{X_j\}=\frac{1}{1-P_{\rm fail}}$ and $\mathcal{E}\{X_j^2\}=\frac{1+P_{\rm fail}}{\left(1-P_{\rm fail}\right)^2}$.
The expressions of $\mathcal{E}\{X_j\}$ and $\mathcal{E}\{X_j^2\}$ can be used to evaluate $\mathcal{E}\{Y_j\}$ and $\mathcal{E}\{Y_j^2\}$, since  $
\mathcal{E}\{Y_j\} =
T N   \mathcal{E}\{X_{j}\}  
$ and   $
\mathcal{E}\{Y_j^2\} =
   N^2T^2 \mathcal{E}\left\{X_j^2 \right\} 
    + 2\mathcal{E}\left\{S_{j}^2 \right\}    - 2  \mathcal{E}\left\{S_{j}\right\}^2$.       
 Furthermore,  $\mathcal{E}\{S_{j}\}$, $\mathcal{E}\{S_{j} ^2 \}$, and  $ \mathcal{E}\{Y_j\} $ can be used to evaluate    $\mathcal{E}\{S_{j-1}Y_j\}$ which can be expressed as follows: 
\begin{align}\nonumber 
\mathcal{E}\{S_{j-1}Y_j\}  
=&\mathcal{E}\{S_{j}\}  \mathcal{E}\{Y_j\} -\mathcal{E}\{S_{j} ^2 \}+\mathcal{E}\{S_{j}\}  ^2  ,
\end{align}
where the last step follows by the fact that $S_{j-1}$ and $Y_j-(NT-S_{j-1})$ are independent.  

As discussed above, the crucial step to evaluate the AoI is to find $\mathbf{P}$ and $\mathbf{p}$, which depends  on the used multiple  access schemes.  

\subsection{OMA-Based Grant-Free Transmission}
The state transition probabilities for the OMA case can be straightforwardly obtained, as shown in the following. 
With OMA, a single user can be served in each time slot, which means that the number of successful users after one time slot can be increased by one at most. Therefore, most of the state transition probabilities in   matrix $\mathbf{P}$ are zero, except for $P_{j,j}$,  $P_{j,j+1}$, and $P_{j,M}$, $0\leq j \leq M-1$. In particular, $P_{j,j}$ denotes the probability of the event that no user succeeds, and is given by \cite{9695972}
\begin{align}
P_{j,j} =1- (M-j)\mathbb{P}_{\rm TX}e^{-\frac{\epsilon}{P}}\left( 1-\mathbb{P}_{\rm TX}\right)^{M-j-1},
\end{align}
where   $\epsilon=2^R-1$.
$P_{j,j+1}$ denotes the probability of the event that a single user, but not  ${\rm U}_1$, succeeds and is given by 
\begin{align}
P_{j,j+1} = (M-j-1)\mathbb{P}_{\rm TX}e^{-\frac{\epsilon}{P}}\left( 1-\mathbb{P}_{\rm TX}\right)^{M-j-1}.
\end{align}
 Furthermore, the $j$-th element of $\mathbf{p}$, denoted by  $P_{j,M}$, is given by
\begin{align}
P_{j,M} = \mathbb{P}_{\rm TX}e^{-\frac{\epsilon}{P}} \left(1-\mathbb{P}_{\rm TX}\right)^{M-j-1}. 
\end{align}

 \subsection{NOMA-Based Grant-Free Transmission}
The benefit of using NOMA is that more users can be admitted simultaneously than for OMA. In particular, with NOMA,  the number of successful users after one time slot can be increased by   $K$ at most, whereas the number of successful users was no more than $1$ for OMA.  This means that     the non-zero state transition probabilities in   matrix $\mathbf{P}$ include $P_{j,j}$,  $P_{j,j+i}$, and $P_{j,M}$, $0\leq j \leq M-1$,  $1\leq i \leq K$ and $j+i\leq M-1$. 

The analysis of the   state transition probabilities for the NOMA case is more challenging than that for the OMA case, mainly due to the application  of SIC. For example, a collision at SNR level $P_k$ can prevent  all those users, which choose  SNR level $P_i$, $i>k$, from being successful. The following lemma provides a high-SNR approximation for the state transition probabilities.

\begin{lemma}\label{lemma1}
At high SNR, the state transition probability, $P_{j,j}$, $0\leq j\leq M-1$, can be approximated as follows: 
\begin{align}
P_{j,j} \approx& 1-\sum^{M-j}_{m=1}{M-j \choose m} \mathbb{P}_{\rm TX} ^{m} (1- \mathbb{P}_{\rm TX} )^{M-j-m}
\\\nonumber &\times \sum_{k=1}^{K}  m\mathbb{P}_{K} \left(   1-k\mathbb{P}_{K}\right) ^{m-1} ,
\end{align}
the state transition probability, $P_{j,j+1}$,  $0\leq j\leq M-2$,  can be approximated as follows: 
\begin{align}\nonumber 
P_{j,j+1} \approx&  (M-j)\mathbb{P}_{\rm TX}  (1- \mathbb{P}_{\rm TX} )^{M-j-1} \frac{M-j-1}{M-j}K\mathbb{P}_{K}
\\  &+\sum^{M-j}_{m=2}{M-j \choose m} \mathbb{P}_{\rm TX} ^{m} (1- \mathbb{P}_{\rm TX} )^{M-j-m}   \sum^{K-1}_{k=1} \frac{M-j-1}{M-j} m     
 \mathbb{P}_{K}
\\\nonumber &\left[
(1-k\mathbb{P}_{K} 
)^{m-1} - \sum^{K}_{\kappa=k+1} (m-1)\mathbb{P}_{K}  (1- \kappa\mathbb{P}_{K} 
)^{m-2} 
\right],
\end{align} 
and the state transition probability, $P_{j,j+i}$, $0\leq j\leq M-3$ and $2\leq i\leq \min\{M-1-j,K\}$, can be approximated as follows: 
\begin{align}\nonumber
P_{j,j+i}  &\approx  {M-j \choose i} \mathbb{P}_{\rm TX} ^{i} (1- \mathbb{P}_{\rm TX} )^{M-j-i}  \sum^{K-i+1}_{k_1=1}\sum^{K}_{k_2=k_1+i-1} \frac{M-j-i}{M-j}   \mathbb{P}_{K}^i {k_2-k_1-1 \choose i-2}\prod^{i}_{p=1}  p
 \\\nonumber &+\sum^{M-j}_{m=i+1}{M-j \choose m} \mathbb{P}_{\rm TX} ^{m} (1- \mathbb{P}_{\rm TX} )^{M-j-m}\\\nonumber&\times   \sum^{K-i}_{k_1=1}\sum^{K-1}_{k_2=k_1+i-1} \frac{M-j-i}{M-j}   \mathbb{P}_{K}^i {k_2-k_1-1 \choose i-2}\prod^{i-1}_{p=0}   (m-p)
\\  &\times \left[
(1-k_2\mathbb{P}_{K} 
)^{m-i} - \sum^{K}_{\kappa=k_2+1} (m-i)\mathbb{P}_{K}  (1- \kappa\mathbb{P}_{K} 
)^{m-i-1} 
\right] .
\end{align} 
\end{lemma} 
\begin{proof}
See Appendix \ref{proof1}.
\end{proof}
Once the transition probability matrix $\mathbf{P}$ is obtained, the elements of $\mathbf{p}$ can be obtained straightforwardly by applying   $\mathbf{P}\mathbf{1}+\mathbf{p}=\mathbf{1}$, where recall that   $\mathbf{1}$ denotes an $M\times 1$ all-one vector.

The closed-form analytical expressions  shown in Lemma \ref{lemma1} allow the evaluation of   the impact of  NOMA on the AoI   without carrying out intensive Monte Carlo simulations. However, the expressions of the state transition probabilities shown in Lemma \ref{lemma1} are quite involved, which  makes it difficult to obtain insights   about the performance difference between OMA and NOMA. For this reason, the special case of $K=2$ and $N=1$ is focused on in the remainder of this section.   $K=2$ means that there are two SNR levels, i.e., the base station needs to carry out two-stage SIC only, which is an important   case  in practice due to its low system complexity. $N=1$ implies  that there is one time slot in each frame, i.e., in each time slot, all   users have updates to   deliver  and hence always participate in    contention. 
 
For this case, the  following lemma provides the optimal choice for the transmission probability $\mathbb{P}_{\rm TX}$.

\begin{lemma}\label{lemma2}
For the special case of $K=2$ and $N=1$,  the optimal choice for the transmission probability $\mathbb{P}_{\rm TX}$ is given by
\begin{align} 
\mathbb{P}_{\rm TX}^*=\frac{\eta}{M}, 
\end{align}
for $M\rightarrow \infty$ and $P\rightarrow \infty$,
where $\eta$ is the root of the following equation: $\left(1-\frac{\eta}{2}\right)e^{-\frac{\eta}{2}}  +\left(1-\frac{\eta^2}{2}\right) e^{-\eta}=0$.
\end{lemma} 
\begin{proof}
See Appendix \ref{proof2}.
\end{proof}

{\it Remark 3:} Note that $\left(1-\frac{\eta}{2}\right)e^{-\frac{\eta}{2}}  +\left(1-\frac{\eta^2}{2}\right) e^{-\eta}=0$ is not related to $M$, which means that $\eta$ is not a function of $M$. By applying   off-shelf root solvers, the exact value of $\eta$ can be straightforwardly obtained  as follows: $\eta\approx 1.6646$. 

By using Lemma \ref{lemma2}, the AoI performance difference between NOMA and OMA is analyzed   in the following proposition. 
\begin{proposition}\label{proposition}
For the special case of $K=2$ and $N=1$, for $M\rightarrow \infty$ and $P\rightarrow \infty$,   the ratio between the AoI achieved by NOMA and OMA is given by
\begin{align}
 \frac{\bar{\Delta}^N }{\bar{\Delta}^O }\approx  \frac{ 2e^{\eta-1} }{ \eta\left( e^{\frac{\eta}{2}} +
   1+\frac{\eta}{2} \right)    } ,
\end{align}
where $\bar{\Delta}^N$ and $\bar{\Delta}^O$ denote the AoI achieved by NOMA and OMA, respectively.
\end{proposition}
\begin{proof}
See Appendix \ref{proof3}.
\end{proof}

{\it Remark 4:} Recall that $\eta\approx 1.6646$, which means that the ratio   in Proposition \ref{proposition} is $\frac{\bar{\Delta}^N }{\bar{\Delta}^O }\approx 0.5653 $, i.e., the use of NOMA can almost halve   the AoI achieved by  OMA. We note that this significant performance gain is achieved by using NOMA with two SNR levels only, i.e., the simplest form of NOMA. By implementing NOMA with more than two SNR levels, the performance gain of NOMA over OMA can be further increased, as shown in the next section.
  
{\it Remark 5:} As shown in the proof for Proposition \ref{proposition}, the optimal choice of $\mathbb{P}_{\rm TX}$ for OMA is $\frac{1}{M}$. This is expected as explained in the following. With $M$ users   competing  for   access in  a single channel, the use of a transmission probability of $\frac{1}{M}$ is  reasonable.  By using the same rationale, one might expect that $\frac{2}{M}$ should be optimal for the NOMA transmission probability  with two SNR levels. However, Lemma \ref{lemma2} shows that the optimal value  of $\mathbb{P}_{\rm TX}$ is $\frac{1.6646}{M}$, which is a more conservative choice for transmission than $\frac{2}{M}$. The reason for this are potential SIC errors. In particular, although there are two channels (or two SNR levels, $P_1$ and $P_2$), a collision at $P_1$ causes SIC to immediately  terminate, which means that $P_2$ can no longer be used to serve  any   users, i.e., the number of the effective channels is less than $2$. 

{\it Remark 6:} Motivated by the results shown in Lemma \ref{lemma2} and Proposition \ref{proposition}, for the general case of $K>2$,  a simple choice of $\mathbb{P}_{\rm TX}=\min\left\{1,\frac{K}{M}\right\}$ can be used for NOMA. In fact,   the simulation results presented in the next section show that this choice is sufficient to realize   a significant performance gain of NOMA over OMA. We note that this choice of $\mathbb{P}_{\rm TX}$ depends on $M$ only, unlike the    the  state-dependent  choice of $\mathbb{P}_{\rm TX}$ used for OMA, which is $\mathbb{P}_{\rm TX}=\frac{1}{M-j}$  \cite{9695972}. Therefore,    an important direction for future research is to find a more sophisticated state-dependent  choice of $\mathbb{P}_{\rm TX}$ for NOMA-assisted grant-free transmission. 

{\it Remark 7:} Proposition \ref{proposition}  implies  that the application of NOMA is particularly beneficial for grant-free transmission, compared to its application to grant-based transmission. 
Recall that in grant-based networks, one important benefit of using NOMA for AoI reduction is that a user can be scheduled to transmit earlier than with OMA \cite{crnomaaoi}. However, for a user which has already been scheduled to transmit early in OMA,   the impact of NOMA on the user' AoI   can be insignificant, particularly at high SNR.  Unlike grant-based networks,   Proposition \ref{proposition}  shows that in grant-free networks, the use of NOMA can  reduce   the AoI of OMA by more than $40\%$, and this significant performance gain applies for  all   users in the network. 

 \begin{figure}[t] \vspace{-2em}
\begin{center}
\subfigure[With Two SNR Levels ($K=2$)  ]{\label{fig3a}\includegraphics[width=0.47\textwidth]{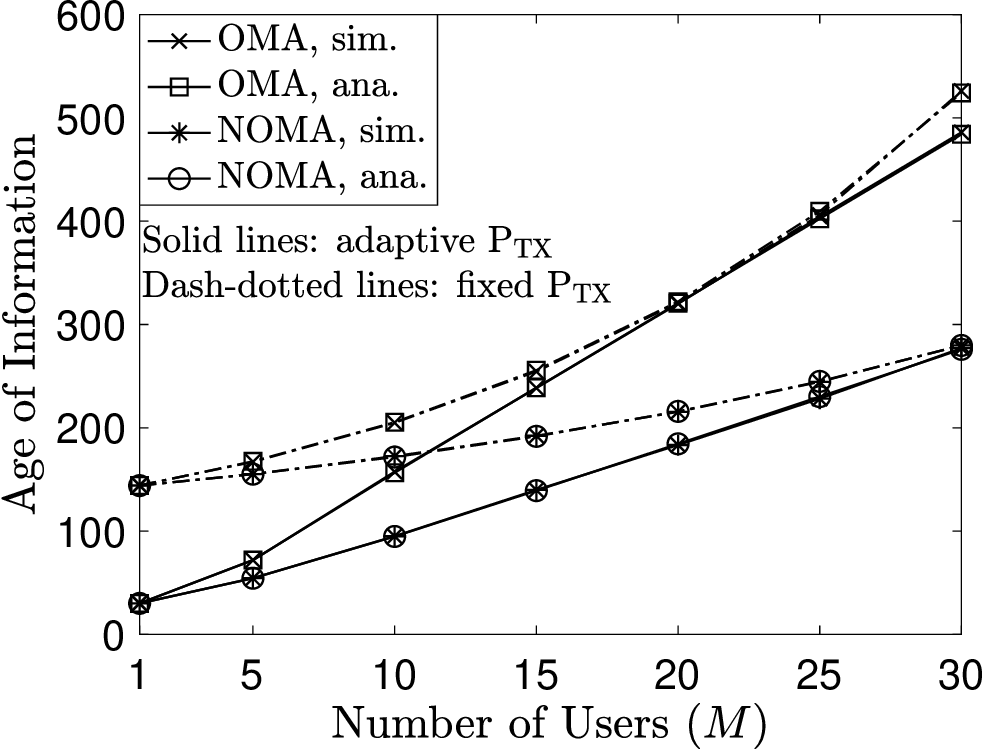}}\hspace{2em}
\subfigure[With Four SNR Levels ($K=4$)  ]{\label{fig3b}\includegraphics[width=0.47\textwidth]{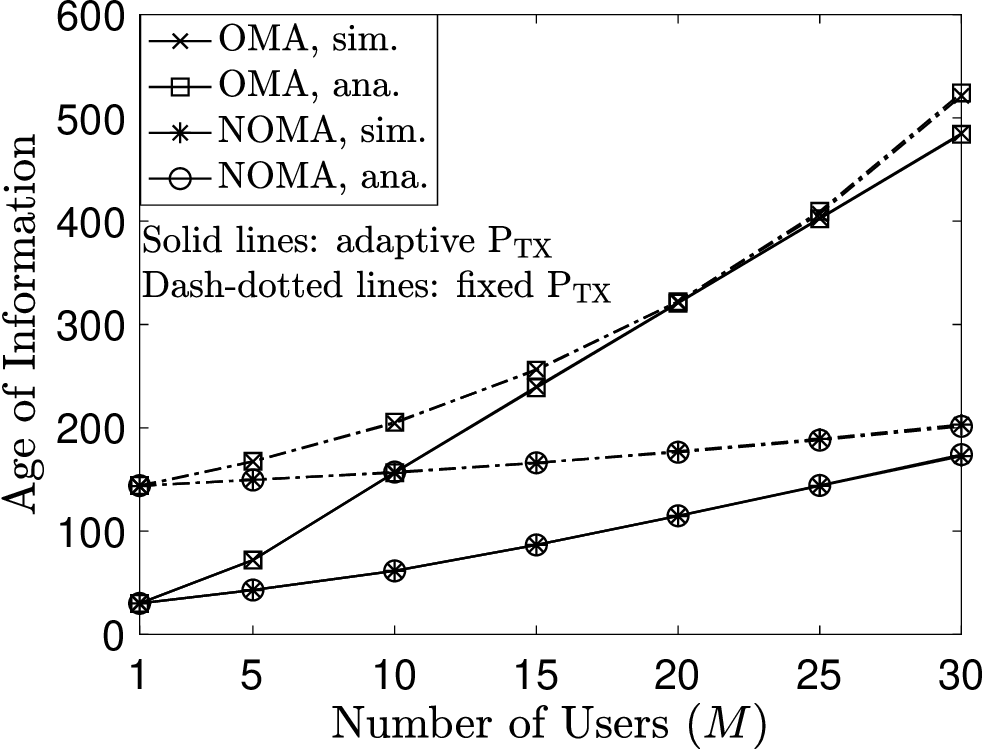}} \vspace{-1.5em}
\end{center}
\caption{Impact of the number of users on the average AoI achieved with OMA and NOMA assisted grant-free transmission for GAR.     $T=6$, $P=20$ dB, $R=0.5$ BPCU, and $N=8$. For the   fixed choice of $\mathbb{P}_{\rm TX}$,  $\mathbb{P}_{\rm TX}=0.05$,  and for the  adaptive choice of $\mathbb{P}_{\rm TX}$,     $\mathbb{P}_{\rm TX}=\min\left\{1,\frac{K}{M}\right\}$ for NOMA, and   $\mathbb{P}_{\rm TX}=\frac{1}{M-j}$ for OMA, i.e., a state-dependent choice is used for OMA as discussed in Section \ref{subsectionx1}. \vspace{-1em} }\label{fig3}\vspace{-1em}
\end{figure}

\section{Simulation Results} \label{section 4}
In this section,     simulation results are presented to demonstrate the AoI achieved by the considered grant-free transmission schemes and to also verify the developed analytical results. 

In Fig. \ref{fig3}, the impact of the number of users on the average AoI achieved by the considered grant-free transmission schemes is investigated. As can be seen from the figure, the AoI achieved with  NOMA is  significantly lower than that of OMA.  In addition, Fig. \ref{fig3} demonstrates that the performance gain of NOMA over OMA increases as the number of users, $M$, grows. This observation can be explained by using  Proposition \ref{proposition} which states that, for $K=2$, $N=1$, $M\rightarrow \infty$ and $P\rightarrow \infty$,  $  \bar{\Delta}^N   \approx  0.5653 \bar{\Delta}^O$, or equivalently $ \bar{\Delta}^N  -\bar{\Delta}^O \approx  0.4347 \bar{\Delta}^O$. Because    increasing $M$ increases   $\bar{\Delta}^O$,    the performance gain of NOMA over OMA also increases as the number of users grows. Therefore, the use of NOMA is particularly important for grant-free transmission with a massive number of users, an important use case for 6G.     Between the two choices of $\mathbb{P}_{\rm TX}$, the adaptive choice yields a better AoI  than the fixed choice. For the two subfigures in Fig. \ref{fig3}, different   numbers of SNR levels, $K$, are used. By comparing the two subfigures, one can observe that the AoI achieved by the NOMA scheme can be   reduced by increasing the number of   SNR levels. This is because the use of more SNR levels makes user collisions less likely to happen.   Fig. \ref{fig3} also demonstrates the accuracy of the AoI expressions developed in Lemma \ref{lemma1}.

    \begin{figure}[t]\centering \vspace{-2em}
    \epsfig{file=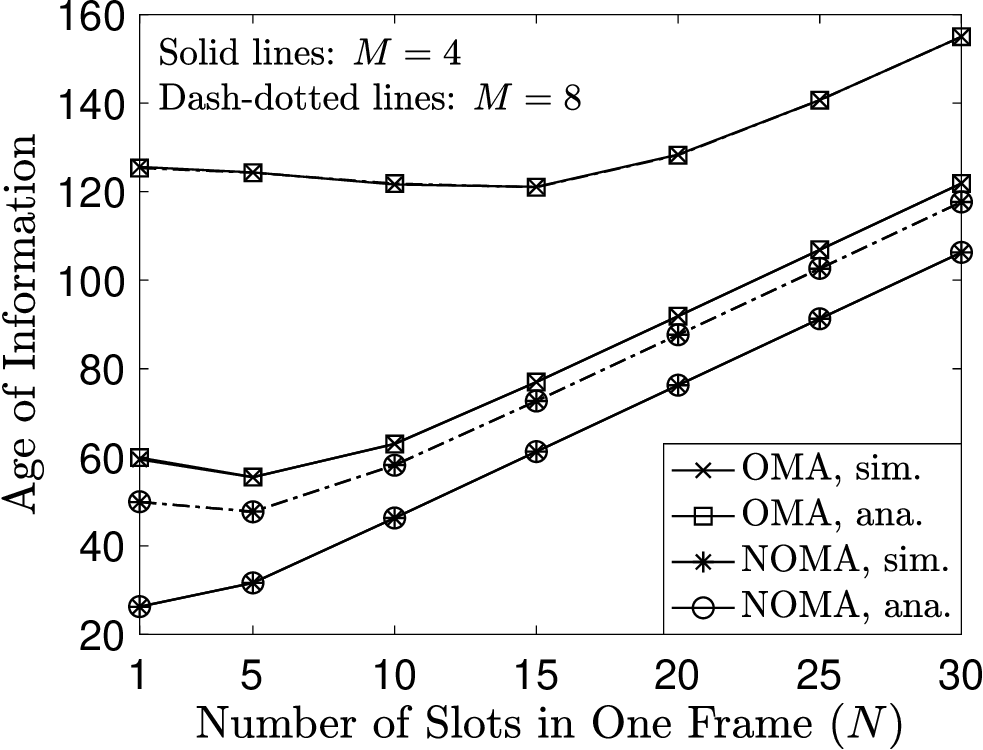, width=0.5\textwidth, clip=}\vspace{-0.5em}
\caption{Impact of the number of time slots in each frame on the average AoI achieved by the OMA and NOMA assisted grant-free transmission schemes  for GAR.  $K=4$, $T=6$, $P=20$ dB, $R=0.5$ BPCU, $M=8$,  and the adaptive choices for $\mathbb{P}_{\rm TX}$ are used.    \vspace{-1em}    }\label{fig4}   \vspace{-0.5em} 
\end{figure}

In Fig. \ref{fig4}, the impact of the number of time slots in each frame, $N$, on the average AoI achieved by the two considered grant-free transmission schemes is studied. As can be seen from the figure, the use of NOMA can always  realize lower AoI than   OMA, regardless of  the choices of $N$. 
An interesting observation from Fig. \ref{fig4} is that a small increase of $N$, e.g., from $1$ to $5$, can reduce the AoI. This is because the likelihood for users to deliver their updates to the base station is improved if there are more time slots in each frame. However,  after   $N\geq 10$, further adding more time slots in each frame   increases the AoI, which can be explained with the following example. Assume that ${\rm U}_1$ can always  successfully update its base station in the first time slot of each frame. For this example, ${\rm U}_1$'s AoI is simply the length of one time frame, and hence its AoI is increased if there are more time slots in one frame. 



 \begin{figure}[t] \vspace{-0em}
\begin{center}
\subfigure[$R=0.1$ BPCU  ]{\label{fig5a}\includegraphics[width=0.47\textwidth]{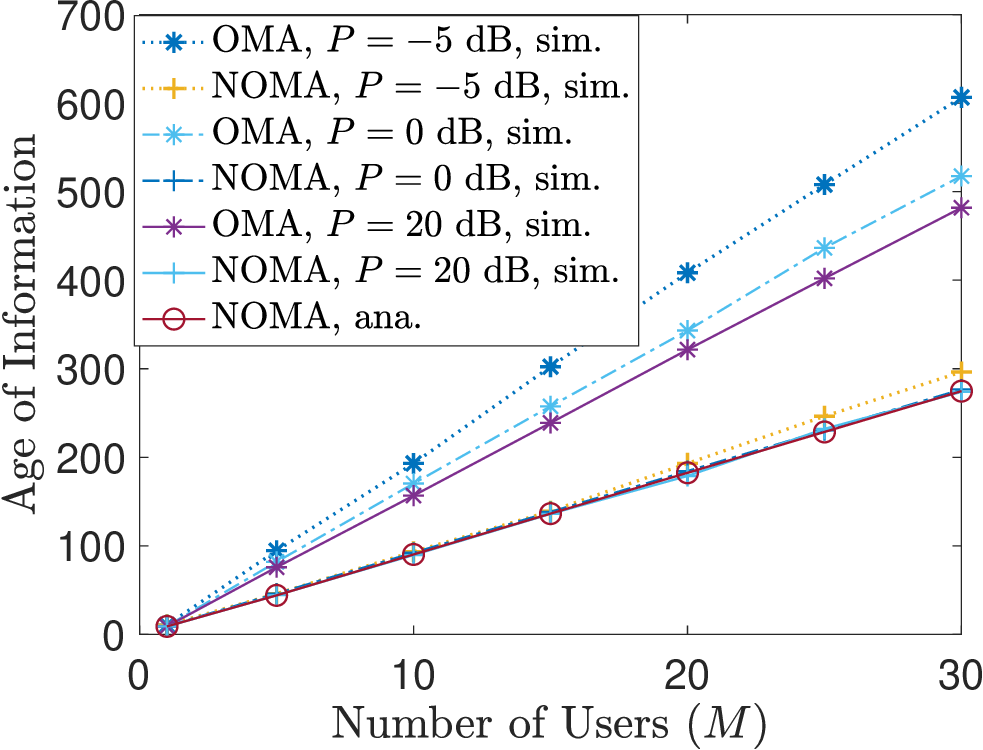}}\hspace{2em}
\subfigure[$R=0.5$ BPCU  ]{\label{fig5b}\includegraphics[width=0.47\textwidth]{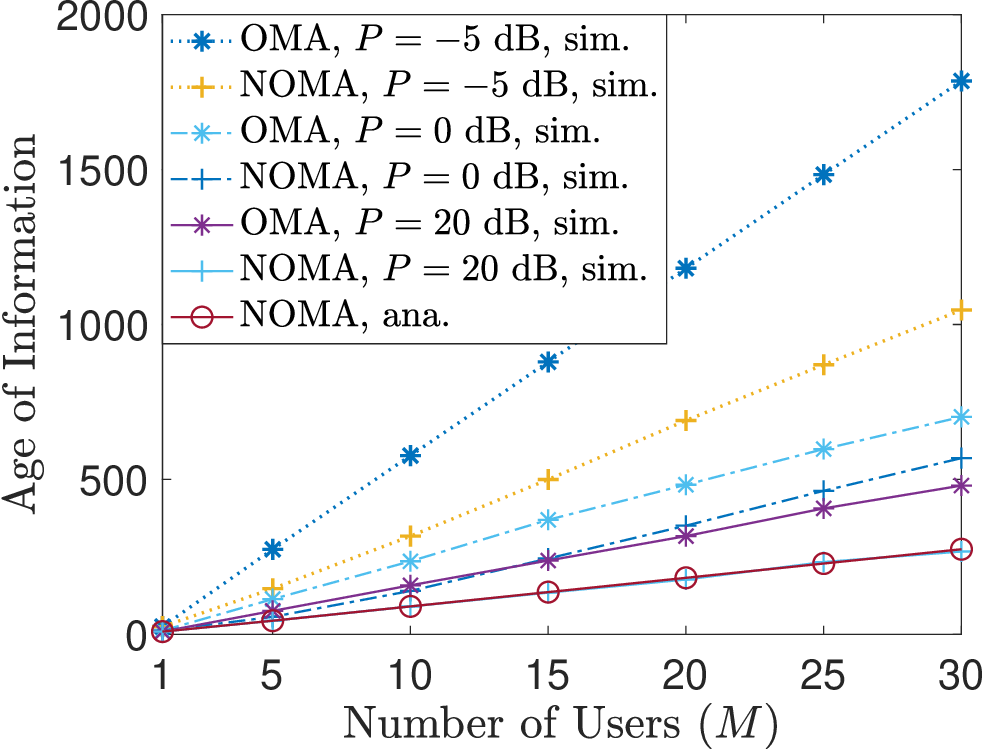}} \vspace{-1.5em}
\end{center}
\caption{AoI performance achieved by the two grant-free transmission schemes for the special case with $K=2$ and $N=1$, where   $T=6$,     and the optimal choices of $\mathbb{P}_{\rm TX}$ are used.        \vspace{-1em}    }\label{fig5}   \vspace{-1em} 
\end{figure}

     \begin{figure}[t]\centering \vspace{-2em}
    \epsfig{file=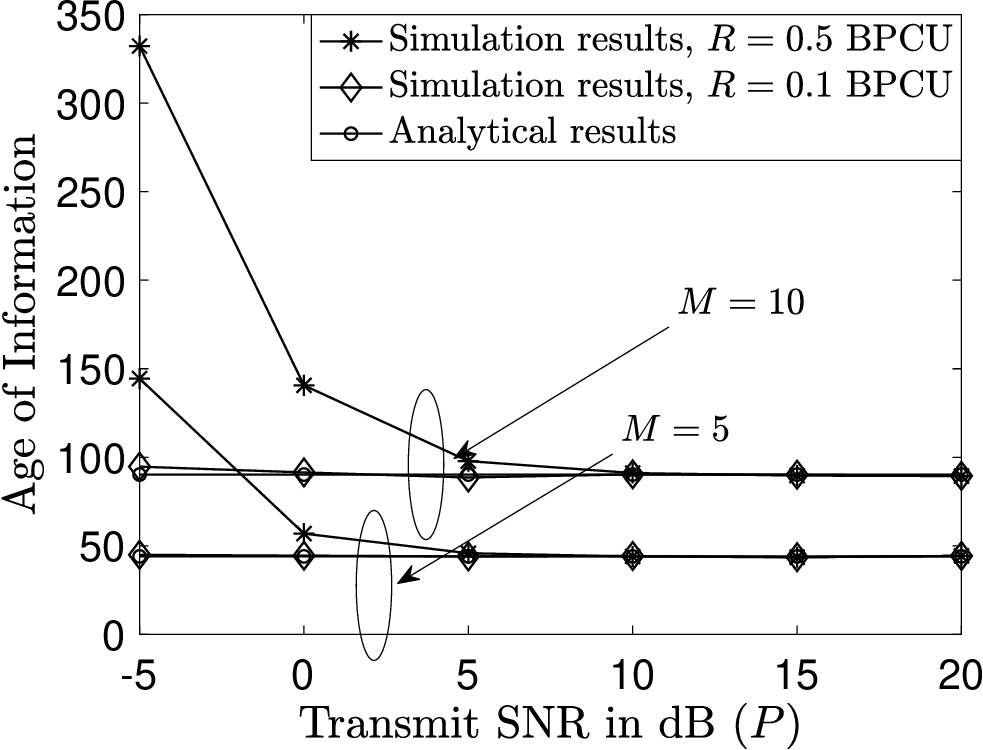, width=0.5\textwidth, clip=}\vspace{-0.5em}
\caption{Illustration of the accuracy of the developed high-SNR analytical results. $K=2$,  $N=1$,     $T=6$,     and the optimal choices of $\mathbb{P}_{\rm TX}$ are used.      \vspace{-1em}    }\label{fig3xx}   \vspace{-1em} 
\end{figure}

As discussed in the previous section, the special case with $K=2$ and $N=1$ is important in practice, and hence the AoI realized by the OMA and NOMA assisted grant-free transmission schemes is investigated in Fig. \ref{fig5}. 
In particular, the figure   shows that the performance gain of NOMA over OMA is particularly large at low SNR. This is a valuable property in practice since most AoI sensitive applications, such as IoT and sensor networks, are  energy constrained and operate in the low SNR regime.    Fig. \ref{fig5a} also demonstrates that for the case with small $R$, even if the   SNR is low, i.e., $-5$ and $0$ dB, the difference between the analytical and simulation results is negligible. This property of the developed analytical results is particularly important, given the fact that, for many important applications of grant-free transmission, such as IoT and uMTC, the users' target data rates are indeed small. 
The aforementioned conclusions are also      confirmed by Fig. \ref{fig3xx}, where the AoI is shown as a function of the transmit SNR. In particular, the developed  analytical results provide   accurate estimates in the medium-to-high SNR regions regardless of the   choices of $R$.

%

 In Fig. \ref{fig6}, the AoI of grant-free transmission is shown as a function of    the transmission probability, $\mathbb{P}_{\rm TX}$, and the figure  demonstrates that  the choices of   $\mathbb{P}_{\rm TX}$ are crucial to the AoI performance of     grant-free transmission. Furthermore, Fig. \ref{fig6} shows that NOMA assisted grant-free transmission always yields a smaller AoI than the OMA case, if the same values for the  transmission probability are used for both of the schemes.    As shown in Lemma \ref{lemma2} and the proof for Proposition \ref{proposition}, 
$\mathbb{P}_{\rm TX}=\frac{1}{M}$ is optimal for OMA, and $\mathbb{P}_{\rm TX}=\frac{\eta}{M}$ is optimal for NOMA in the case with $K=2$ and $N=1$.   Fig. \ref{fig6} verifies the optimality of these choices of $\mathbb{P}_{\rm TX}$, since   the minimal AoIs achieved by the fixed choices of $\mathbb{P}_{\rm TX}$ match perfectly with the AoIs realized by the optimal choices of $\mathbb{P}_{\rm TX}$. 

  \begin{figure}[t]\centering \vspace{-2em}
    \epsfig{file=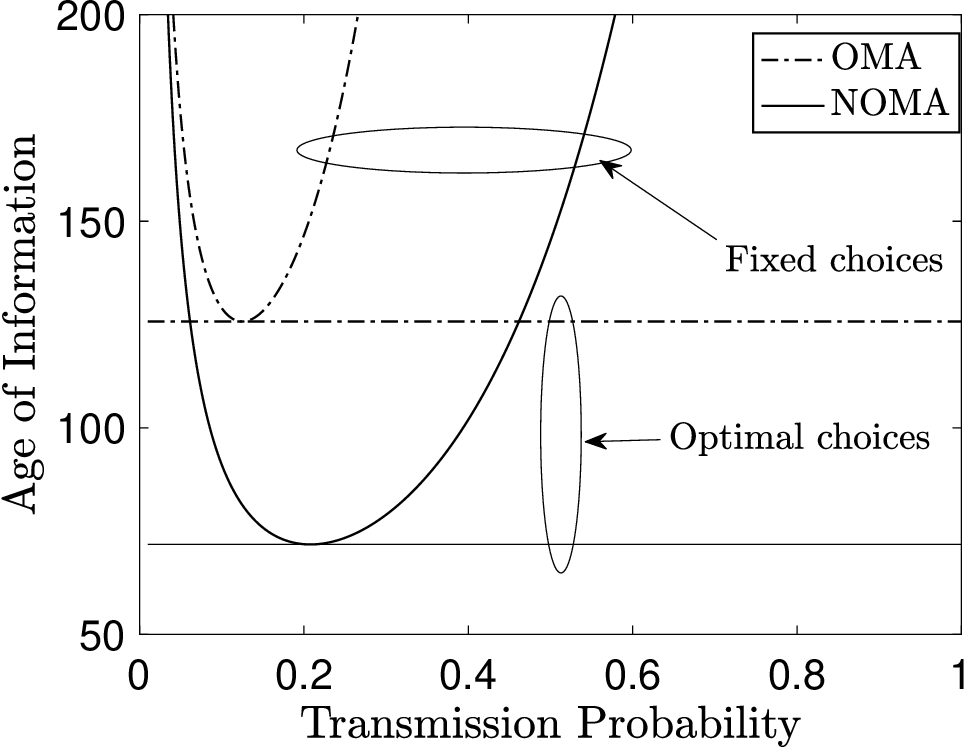, width=0.5\textwidth, clip=}\vspace{-0.5em}
\caption{Illustration of the impact of $\mathbb{P}_{\rm TX}$ on the AoI achieved by the considered grant-free schemes  for GAR.  $T=6$, $P=20$ dB, $R=0.5$ BPCU, $K=2$, $N=1$, and $M=8$.   \vspace{-1em}    }\label{fig6}   \vspace{-0.2em} 
\end{figure}
     \begin{figure}[t]\centering \vspace{-0em}
    \epsfig{file=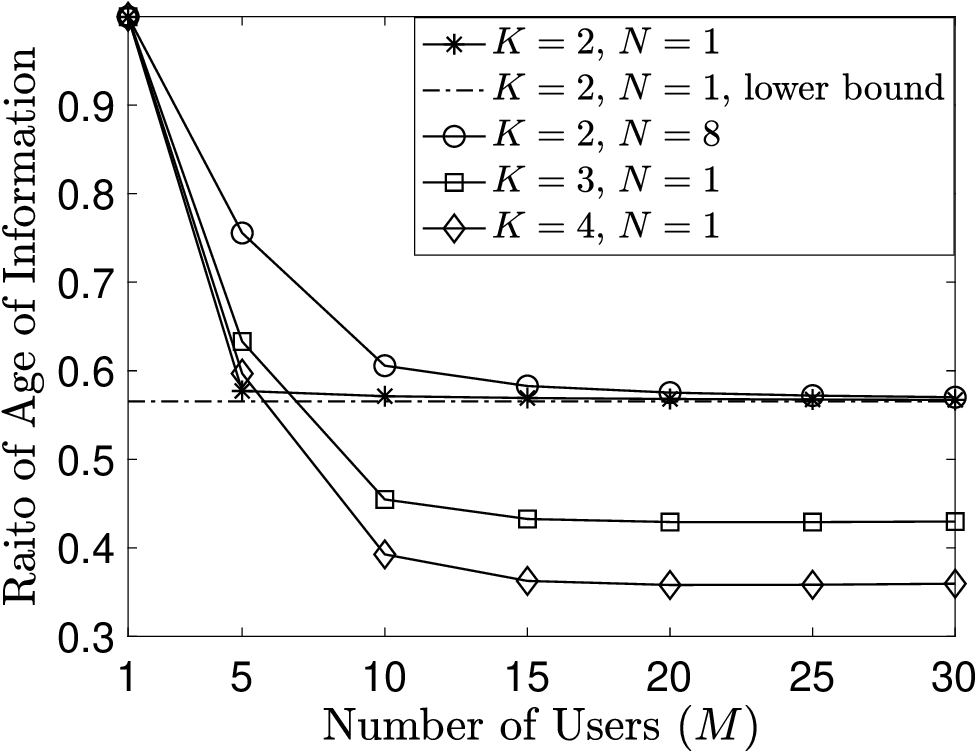, width=0.5\textwidth, clip=}\vspace{-0.5em}
\caption{Impact of the number of users on the ratio between the AoI achieved with the NOMA and OMA, i.e., $\frac{\bar{\Delta}^N}{\bar{\Delta}^O}$,  for GAR.    $T=6$, $P=20$ dB, $R=0.5$ BPCU,  and the optimal choices of $\mathbb{P}_{\rm TX}$ are used.       \vspace{-1em}    }\label{fig7}   \vspace{-0.1em} 
\end{figure}

In Fig. \ref{fig7}, the performance of the considered OMA and NOMA grant-free schemes are compared  by using the following AoI ratio, $\frac{\bar{\Delta}^N}{\bar{\Delta}^O}$. For the special case of $N=1$ and $K=2$, Proposition \ref{proposition} predicts that this ratio is $0.5653$ for large $M$, which is confirmed by Fig. \ref{fig7}. If    $K$ is fixed,   i.e.,  $K=2$, an increase of $N$ does not change the ratio significantly, particularly in the case of large $M$.  By introducing more SNR levels, i.e., increasing $K$, the AoI ratio can be further reduced, which means that the performance gain of NOMA over OMA can be increased by introducing more SNR levels.  This is expected since increasing $K$ reduces the likelihood of user collisions and ensures that users can update the base station earlier. 

     \begin{figure}[t]\centering \vspace{-2em}
    \epsfig{file=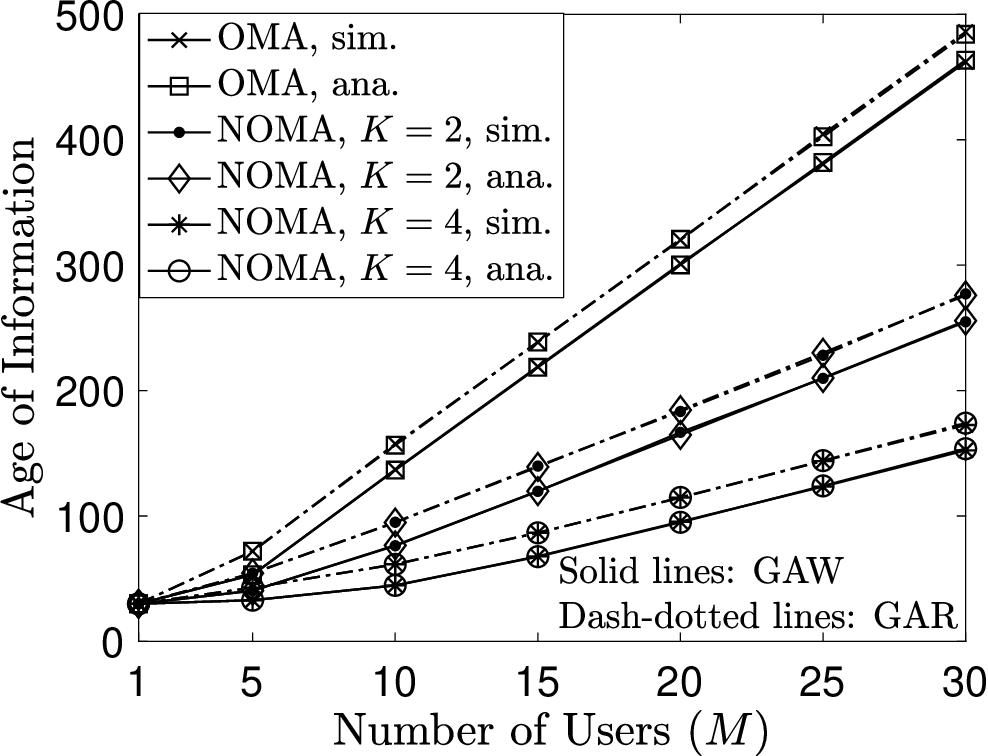, width=0.5\textwidth, clip=}\vspace{-0.5em}
\caption{AoI achieved by the considered grant-free transmission schemes for different data generation models.  $N=8$,    $T=6$,   $R=0.5$ BPCU, and the adaptive choices of $\mathbb{P}_{\rm TX}$ are used.         \vspace{-1em}    }\label{fig8}   \vspace{-1em} 
\end{figure}

For all  the previous simulation results, GAR has been considered, which means that a user's update is generated at the beginning of a time frame, instead of at the beginning of a time slot as for GAW. 
In Fig. \ref{fig8}, the AoI achieved by the considered grant-free transmission schemes for the two different data generation models is illustrated. As can be seen from the figure, for both GAR and GAW, NOMA assisted grant-free transmission always outperforms the OMA based scheme. In addition, the figure shows that the AoI realized by the considered schemes for GAW is smaller than that for GAR, because, for GAW, each update is generated right before its delivery time, i.e., there is no service delay $S_j$. We also note that  the difference between the AoI for GAW and GAR is not significant, but the use of GAW can cause a higher  energy consumption  than   GAR, since GAW requires a user to re-generate an update for each retransmission.

 \section{Conclusions}\label{section 5}
In this paper,   the impact of NOMA on the AoI of grant-free transmission has been investigated  by applying   a particular form of NOMA, namely NOMA-assisted random access.    By modelling grant-free transmission  as a Markov chain and accounting  for SIC, closed-form analytical expressions  for the AoI achieved by NOMA assisted grant-free transmission have been  obtained, and asymptotic studies have been  carried out to demonstrate that the use of the simplest form of NOMA is already sufficient to   reduce   the AoI of OMA by more than $40\%$. In addition, the developed analytical results have  also been shown  useful  for optimizing  the users' transmission probabilities, $\mathbb{P}_{\rm TX}$, which is crucial for   performance maximization    of grant-free transmission. 

In this paper, concise and insightful analytical  results have been developed for the special case of $N=1$. An important direction for future research is the development of    similar insightful  results   for the general case of $N\geq 1$.     We also note that the use of NOMA may reduce a user's energy consumption by avoiding the possible large number of retransmissions needed for  OMA; however, a user that chooses  a high receive SNR level, e.g., $P_1$, may consume more energy than in OMA.  Therefore, another important direction for future research is to  study how to realize a balanced   tradeoff between energy efficiency and AoI reduction. 
\section{Acknowledgements}
The authors thank Dr. Jinho Choi for his kind suggestions about the implementation of NOMA assisted random access. 
\vspace{-1em}
\appendices
\section{Proof for Lemma \ref{lemma1}}\label{proof1}

The proof is divided into three parts to evaluate  $P_{j,j}$, $P_{j,j+1}$, and $P_{j,j+i}$, $i\geq2$, respectively. Throughout the proof, the high SNR assumption is made, which ensures that all the SNR levels, $P_k$, $1\leq k \leq K$, are feasible for each user, i.e., transmission failures are   due to user collisions only.   The users'   channel gains in different time slots are assumed to be independent   and identically   complex Gaussian distributed with zero mean and unit variance. 

\subsection{Evaluating  $P_{j,j}$}
To find the expression for  $P_{j,j}$,   assume that $j$ users have successfully delivered their updates to the base station. Therefore, each of the remaining   $M-j$  users independently makes an attempt to transmit with the probability,   $\mathbb{P}_{\rm TX}$, at a randomly chosen SNR level.  
 $P_{j,j} $ is the probability of  the event that  none of the $M-j$   users succeeds. 

Define $E_{P_k|j}$ as the event that given $M-j$ remaining users, a   user  successfully updates the base station by using the $k$-th SNR level, $P_k$, and no user succeeds at $P_i$, $i<k$. The reason to include the constraint that no user succeeds at $P_i$, $i<k$, in the definition of  $E_{P_k|j}$  is to ensure that $E_{P_k|j}$ and $E_{P_p|j}$, $k< p$, are uncorrelated. For example, the event that ${\rm U}_i$ succeeds by using $P_2$ and ${\rm U}_j$ succeeds by using $P_3$ belongs to  $E_{P_2|j}$ only, and is not included in  $E_{P_3|j}$ .

Therefore, the probability $P_{j,j}$ can be expressed as follows: 
\begin{align}
P_{j,j} =& 1-\mathbb{P}\left(E_{P_1|j}\cup\cdots \cup E_{P_K|j}\right)
\\\nonumber =& 1-\sum^{K}_{k=1}\mathbb{P}\left(E_{P_k|j} \right).
\end{align}

Further define $E_{m|j}$ as the event that  among the $M-j$ remaining  users, there are $m$ active users which make the transmission attempts, and define $E_{P_k|m}$ as the event that among the $m$ active users, a single user  successfully updates the base station by using the $k$-th SNR level, $P_k$, and no user chooses $P_i$, $i<k$, which means
\begin{align}\label{pjj 1}
P_{j,j} =&   1-\sum^{K}_{k=1}\sum^{M-j}_{m=1}{M-j \choose m}\mathbb{P}\left(E_{m|j}\right)\mathbb{P}\left(E_{P_k|m} \right).
\end{align} 

 By using the transmission attempt probability of $  \mathbb{P}_{\rm TX}$, the probability, $E_{m|j}$, can be obtained as follows:
\begin{align}\label{pjj 2}
\mathbb{P}(E_{m|j}) =  \mathbb{P}_{\rm TX} ^{m} (1- \mathbb{P}_{\rm TX} )^{M-j-m}.
\end{align}
Without loss of generality, assume that ${\rm U}_i$  is one of the $m$ active users which make transmission attempts.  The probability for the event that ${\rm U}_i$ chooses $P_k$, and    no user chooses $P_i$, $i<k$, is given by
\begin{align}
\mathbb{P}_{K} \left(   1-k\mathbb{P}_{K}\right) ^{m-1} ,
\end{align}
where $ (1-k\mathbb{P}_{K})$ is the probability of the event that   a user which is not 
${\rm U}_i$ cannot choose $P_p$, $1\leq p \leq k$. Therefore, $\mathbb{P}\left(E_{P_k|j,m} \right)$ can be approximated  as follows:
\begin{align}\label{pjj 3}
\mathbb{P}\left(E_{P_k|m} \right) \approx m\mathbb{P}_{K} \left(   1-k\mathbb{P}_{K}\right) ^{m-1} ,
\end{align}
since each of the  $m$ active users can be the successful user with the equal probability, where  the high SNR assumption is used, i.e., all  SNR levels are   feasible to each user and only the errors caused by user collisions are considered.   By combining \eqref{pjj 1}, \eqref{pjj 2}, and \eqref{pjj 3}, $P_{j,j}$ can be expressed as follows:
\begin{align}
P_{j,j} \approx& 1-\sum^{M-j}_{m=1}{M-j \choose m} \mathbb{P}_{\rm TX} ^{m} (1- \mathbb{P}_{\rm TX} )^{M-j-m}
\\\nonumber &\times \sum_{k=1}^{K}  m\mathbb{P}_{K} \left(   1-k\mathbb{P}_{K}\right) ^{m-1} .
\end{align}

\subsection{Evaluating $P_{j,j+1}$ } 
Recall that $P_{j,j+1}$ is also conditioned on the assumption that $j$ users have successfully updated the base station, and  $P_{j,j+1}$  is the probability of the event that   there is a single successful update from a user which cannot be  ${\rm U}_1$.  

Define $E^1_{P_k|j}$ as the event that given $M-j$ remaining  users, a single user, other than  ${\rm U}_1$, successfully updates the base station by using the $k$-th SNR level. At high SNR, the following two conclusions can be made regarding  $E^1_{P_k|j}$.  On the one hand, due to the feature of SIC, $E^1_{P_k|j}$ implies that      no user chooses $P_i$, $i<k$, which can be shown by contradiction. Assume that ${\rm U}_p$ chooses $P_i$. If ${\rm U}_p$ is the only user choosing $P_i$, ${\rm U}_p$ becomes an additional  successful user, which contradicts   the assumption that   there is a single      successful user. If multiple users choose $P_i$, a collision occurs and SIC needs to be terminated at $P_i$, which contradicts   the assumption  that a    successful update  happens at $P_k$.  On the other hand,  $E^1_{P_k|j}$ does not exclude the event that an SNR level, $P_i$, $i>k$, is chosen by a user; however, $E^1_{P_k|j}$ does imply  that  if  $P_i$, $i>k$, is chosen, a collision must happen at this SNR level, otherwise there will be an additional successful user.   We also note that $E^1_{P_k|j}$ is different from $E_{P_k|j}$ for the following two reasons. First,  $E_{P_k|j}$ does not exclude the event that the successful user is ${\rm U}_1$. Second,    for $E_{P_k|j}$, it is still possible for a user to succeed at   SNR levels, $P_i$, $i>k$.

By using $E^1_{P_k|j}$, the probability $P_{j,j+1}$ can be expressed as follows:  
\begin{align}
P_{j,j+1} =& \mathbb{P}(E^1_{L_1|j}\cup\cdots \cup E^1_{L_K|j})
= \sum^{K}_{k=1}\mathbb{P}(E^1_{L_k|j} ),
\end{align} 
where the last step follows by the fact that the $E^1_{P_k|j}$, $1\leq k \leq K$, are uncorrelated.  

Similar to $E_{P_k|m}$, define $E^1_{P_k|m}$ as the event that among the $m$ active users,   a single user, other than ${\rm U}_1$, successfully updates the base station by using the $k$-th SNR level. By using  $E_{m|j}$ and $E^1_{P_k|m}$, $P_{j,j+1}$ can be expressed as follows:
\begin{align}\label{pjj1 1}
P_{j,j+1} =&   \sum^{K}_{k=1}\sum^{M-j}_{m=1}{M-j \choose m}\mathbb{P}\left(E_{m|j}\right)\mathbb{P}\left(E_{P^1_k|m} \right).
\end{align} 

The analysis of $\mathbb{P}\left(E_{P^1_k|m}\right)$ is challenging. First, we assume that $m\geq 2$, i.e., there are more than one active users. For illustrative purposes, assume that ${\rm U}_2$ is   an active user.  Denote by $E^+_{2k}$ the event that   ${\rm U}_2$ succeeds by using $P_k$, and $\mathbb{P}\left( E^+_{2k}\right)$ is given by
\begin{align}
\mathbb{P}\left( E^+_{2k}\right) \approx  \mathbb{P}_{K}(1-k\mathbb{P}_{K} 
)^{m-1},
\end{align}
where the high SNR approximation is used, and  the reason to have $ (1- k\mathbb{P}_{K} 
)^{m-1} $ is that the $k$ highest SNR levels are no longer available to the other $m-1$ active users. Because $m\geq 2$, $k\leq K-1$, i.e., ${\rm U}_2$ cannot succeed  by using $P_K$, which can be explained by using a simple example with  ${\rm U}_2$ and ${\rm {U}_3}$ being the active users ($m=2$).  As discussed previously, if ${\rm U}_2$ chooses $P_K$, ${\rm {U}_3}$ has to choose $P_p$, $p>K$, which is not possible.

Furthermore, denote by $E^-_{2k}$ the event that   ${\rm U}_2$ succeeds by using $P_k$, and there is at least one additional user which succeeds by using $P_i$, $i\geq k+1$.  $\mathbb{P}\left( E^-_{2k}\right)$ is given by
 \begin{align}
\mathbb{P}\left( E^-_{2k}\right) \approx  \mathbb{P}_{K}\sum^{K}_{\kappa=k+1} (m-1)\mathbb{P}_{K}  (1- \kappa\mathbb{P}_{K} 
)^{m-2} ,
\end{align}
which is obtained in a similar manner as  $\mathbb{P}\left(E_{P_k|m} \right)$ in \eqref{pjj 3}.

   By using $E^+_{2k}$ and $E^-_{2k}$, the probability for the event that among the $m$ active users, only ${\rm U}_2$ succeeds by using $P_k$, denoted by $E_{2k}$, is given by
 \begin{align}
\mathbb{P}\left( E_{2k}\right)  = & \mathbb{P}\left( E^+_{2k}\right) - \mathbb{P}\left( E^-_{2k}\right)  \approx   \mathbb{P}_{K}\left[(1-k\mathbb{P}_{K} 
)^{m-1}- \sum^{K}_{\kappa=k+1} (m-1)\mathbb{P}_{K}  (1- \kappa\mathbb{P}_{K} 
)^{m-2} \right].
\end{align}   

Intuitively, $\mathbb{P}\left(E_{P^1_k|m}\right)$ should be simply $\mathbb{P}\left(E_{P^1_k|m}\right)=(m-1)\mathbb{P}\left( E_{2k}\right)$, i.e., including $(m-1)$ cases corresponding to ${\rm U}_i$, $2\leq i \leq m$. If  ${\rm U}_1$ is one of the active users, indeed, $\mathbb{P}\left(E_{P^1_k|m}\right)=(m-1)\mathbb{P}\left( E_{2k}\right)$. However, if ${\rm U}_1$ is not an active user, $\mathbb{P}\left(E_{P^1_k|m}\right)=m\mathbb{P}\left( E_{2k}\right)$. Therefore,  by taking into   account the fact that ${\rm U}_1$ might not be an active user, $\mathbb{P}\left(E_{P^1_k|m}\right)$ can be expressed as follows:
\begin{align}\label{pjj1 2}
\mathbb{P}\left(E_{P^1_k|m}\right) \approx & \left[\frac{m}{M-j} (m-1)    +\frac{M-j-m}{M-j} (m)     \right]\mathbb{P}\left( E_{2k}\right)\\ \nonumber
= & \frac{M-j-1}{M-j}m \mathbb{P}_{K}  \left[(1-k\mathbb{P}_{K} 
)^{m-1}- \sum^{K}_{\kappa=k+1} (m-1)\mathbb{P}_{K}  (1- \kappa\mathbb{P}_{K} 
)^{m-2} \right],
\end{align}
for the case $m\geq 2$.

For the case $m=1$, i.e., there is a single active user, $\mathbb{P}\left(E_{P^1_k|m}\right)$ is simply zero if this user is ${\rm U}_1$. Otherwise, $\mathbb{P}\left(E_{P^1_k|m}\right)=\mathbb{P}_{K} $, i.e., $P_k$ is selected by the active user. Therefore, for $m=1$, $\mathbb{P}\left(E_{P^1_k|m}\right)$ is given by
\begin{align}\label{pjj1 3}
\mathbb{P}\left(E_{P^1_k|m}\right) \approx  \frac{M-j-1}{M-j}\mathbb{P}_{K} .
\end{align}

By combining \eqref{pjj1 1}, \eqref{pjj1 2} and \eqref{pjj1 3},   the probability, $P_{j,j+1} $, can be obtained as shown in the lemma.

\subsection{Evaluating $P_{j,j+i} $, $i\geq 2$} 
We note that for $P_{j,j+1} $, there is a single successful update, and $P_{j,j+1} $ has been analyzed by first specifying which SNR level, i.e., $P_k$, is used for this successful update.  The same method could be applied   to analyze $P_{j,j+i} $, $i\geq 2$; however, the resulting expression can be extremely complicated if the number of the used SNR levels is large.

Instead, the feature of SIC can be used to simplify the analysis of   $P_{j,j+i} $. Consider the following example with $i=3$ and $K=8$. Assume that the following three SNR levels, $P_2$, $P_4$, and $P_5$, are used by the successful users.  The key observation for simplifying the performance analysis is that those SNR levels prior to $P_2$ and between the chosen ones are not selected by any users, e.g., no user chooses   $P_1$ and $P_3$. This observation can be explained by taking $P_3$ as an example. If this SNR level has been selected, a collision must happen at this level, otherwise there should be an additional successful user. However, if a collision  does happen at $P_3$, SIC needs to be terminated in the third SIC stage, and hence   no successful update can happen at $P_4$, which contradicts to the assumption that $P_4$ is used by a successful user. As a result,  no one chooses $P_3$. 

By using this observation, we note that only two SNR levels are significant  to the analysis of $P_{j,j+i} $, namely the highest and the lowest SNR levels, which are denoted by $P_{k_1}$ and $P_{k_2}$, respectively. For the aforementioned example, $P_{k_1}=P_2$ and $P_{k_2}=P_5$. 
Define $E^i_{P_k|j}$ as the event that given $M-j$ remaining  users, $i$ users which are not ${\rm U}_1$ successfully update the base station, where $P_{k_1}$ and $P_{k_2}$ are the highest and lowest used SNR levels.  
By using $E^i_{P_k|j}$, the probability $P_{j,j+i}$ can be expressed as follows:  
\begin{align}
P_{j,j+i} =&  \sum^{K-i+1}_{k_1=1}\sum^{K}_{k_2=k_1+i-1}\mathbb{P}(E^i_{L_k|j} )
=  \sum^{K-i+1}_{k_1=1}\sum^{K}_{k_2=k_1+i-1}\sum^{M}_{m=1}{M-j \choose m}\mathbb{P}\left(E_{m|j}\right)\mathbb{P}\left(E_{P^i_k|m} \right),
\end{align} 
where $E_{P^i_k|m} $ is defined similar to $E^i_{P_k|j}$ with the assumption that there are  $m$ active users. 

In order to find $\mathbb{P}\left(E_{P^i_k|m} \right)$, again, we first focus on  the case   $m\geq i+1$, i.e., there are more than $i$ active user.  Define $E_{ik}$ as the probability for the particular event that among the $m$ active users,   ${\rm U}_2$ succeeds by using $P_{k_1}$,   ${\rm U}_j$ succeeds by using $P_{k_1+j-2}$, $3\leq j\leq i$, and ${\rm U}_{i+1}$ succeeds by using $P_{k_2}$. Similar to $\mathbb{P}\left( E_{2k}\right)$,  $\mathbb{P}\left( E_{ik}\right)$  can be approximated at high SNR as follows:  
 \begin{align}\label{eee}
\mathbb{P}\left( E_{ik}\right) \approx \mathbb{P}_{K}^i\left[(1-k_2\mathbb{P}_{K} 
)^{m-i}- \sum^{K}_{\kappa=k_2+1} (m-i)\mathbb{P}_{K}  (1- \kappa\mathbb{P}_{K} 
)^{m-i-1} \right],
\end{align}   
where the term, $(1-k_2\mathbb{P}_{K} 
)^{m-i}$, is due to the fact that the remaining $m-i$ active users can choose $P_i$, $i>k_2$, only, and the second term in the bracket is obtained similar to $\mathbb{P}\left(E_{P_k|m} \right)$ in \eqref{pjj 3}. 

Following   steps similar to those to obtain $\mathbb{P}\left(E_{P^1_k|m}\right) $ in \eqref{pjj1 2},  $\mathbb{P}\left(E_{P^i_k|m} \right)$ can be obtained from $\mathbb{P}\left( E_{ik}\right) $ as follows:
\begin{align}\label{ex1}
\mathbb{P}\left(E_{P^i_k|m} \right) \approx & \mathbb{P}\left( E_{ik}\right){k_2-k_1-1 \choose i-2}  \underset{\text{repetitions if the tagged user is one of the $m$ active users }}{\underbrace{\left[
\frac{m}{M-j} (m-1)   (m-2) \cdots (m-i)    +
\right.}} 
\\\nonumber & \underset{\text{repetitions if the tagged user is not one of the $m$ active users}}{\underbrace{\left.
\frac{M-j-m}{M-j} m(m-1)\cdots (m-i+1)  
\right]}}\\\nonumber
=&\mathbb{P}\left( E_{ik}\right){k_2-k_1-1 \choose i-2} 
 \frac{M-j-i}{M-j} m\cdots (m-i+1)    , 
\end{align}
for the case $m\geq i+1$,
 where ${k_2-k_1-1 \choose i-2}$ is the number of  the possible choices for the $(i-2)$ SNR levels which are between $P_{k_1}$ and $P_{k_2}$. 
 
The special case $m=i$ means that there are $m$ active users and each of the active users is a successful user. Therefore, the probability in \eqref{eee}, $
\mathbb{P}\left( E_{ik}\right)$, is simply $ \mathbb{P}_{K}^i$, and hence for the case $m=i$, $\mathbb{P}\left(E_{P^i_k|m} \right)$ can be expressed as follows:
\begin{align}\label{ex2}
&\mathbb{P}\left(E_{P^i_k|m} \right) \approx  \mathbb{P}_{K}^i
 {k_2-k_1-1 \choose i-2} 
 \frac{M-j-i}{M-j} i\cdots 1    . 
\end{align}
By combining \eqref{ex1} and \eqref{ex2}, the expression of $P_{j,j+i}$ can be obtained as shown in the lemma. Therefore, the proof for the lemma is complete. 

\section{Proof for Lemma \ref{lemma2}} \label{proof2}
The proof comprises  first simplifying   the state transition probabilities, then developing an asymptotic expression of the AoI, and finally finding the optimal choice of $\mathbb{P}_{\rm TX}$. 
\subsection{Simplifying the State Transition Probabilities}
For the case of $N=1$ and $K=2$,  only the following transition probabilities, $P_{0,0} $, $P_{0,1} $, $P_{0,2} $, need to be focused on.

In particular, the expression of $P_{0,0} $ can be simplified as follows:
\begin{align}\nonumber
P_{0,0} \approx& 1-\sum^{M}_{m=1}{M \choose m} \mathbb{P}_{\rm TX} ^{m} (1- \mathbb{P}_{\rm TX} )^{M-m}  \sum_{k=1}^{K}  m\mathbb{P}_{K} \left(   1-k\mathbb{P}_{K}\right) ^{m-1} \\\nonumber
 \overset{(1)}{=}& 1-\sum^{M}_{m=1}{M \choose m} \mathbb{P}_{\rm TX} ^{m} (1- \mathbb{P}_{\rm TX} )^{M-m}   m\mathbb{P}_{K} \left(   1-\mathbb{P}_{K}\right) ^{m-1} -
 M \mathbb{P}_{\rm TX}  (1- \mathbb{P}_{\rm TX} )^{M-1}   \mathbb{P}_{K}  
 \\\nonumber
 \overset{(b)}{=}& 1-\sum^{M}_{m=1}{M \choose m} \mathbb{P}_{\rm TX} ^{m} (1- \mathbb{P}_{\rm TX} )^{M-m}   m\mathbb{P}_{K}^m -
 M \mathbb{P}_{\rm TX}  (1- \mathbb{P}_{\rm TX} )^{M-1}   \mathbb{P}_{K} , 
\end{align}
where   step $(a)$ follows by the fact that $ (1-K\mathbb{P}_{K})^{m-1}\neq 0$ only if $m=1$,  and     step $(b)$  follows by  $\mathbb{P}_{K}=1-\mathbb{P}_{K}$ for $K=2$. By using the properties of the binomial coefficients, $P_{0,0}$ can be simplified as follows:
\begin{align}
P_{0,0} \approx& 1-M\mathbb{P}_{K}\mathbb{P}_{\rm TX}\sum^{M-1}_{i=0}{M-1 \choose i} (\mathbb{P}_{K}\mathbb{P}_{\rm TX}) ^i    (1- \mathbb{P}_{\rm TX} )^{M-1-i} -
 M \mathbb{P}_{\rm TX}  (1- \mathbb{P}_{\rm TX} )^{M-1}   \mathbb{P}_{K}  
\\\nonumber
=& 1-M\mathbb{P}_{K}\mathbb{P}_{\rm TX}   (\mathbb{P}_{K}\mathbb{P}_{\rm TX}+1- \mathbb{P}_{\rm TX} )^{M-1}  -
 M \mathbb{P}_{\rm TX}  (1- \mathbb{P}_{\rm TX} )^{M-1}   \mathbb{P}_{K}  .
\end{align}

Similarly, for the case of $K=2$, $P_{0,1}$ can be first written as follows:
\begin{align}\nonumber
P_{0,1} \approx&   M\mathbb{P}_{\rm TX}  (1- \mathbb{P}_{\rm TX} )^{M-1} \frac{M-1}{M}K\mathbb{P}_{K}
 +\sum^{M}_{m=2}{M \choose m} \mathbb{P}_{\rm TX} ^{m} (1- \mathbb{P}_{\rm TX} )^{M-m}    \sum^{K-1}_{k=1} \frac{M-1}{M} m     
  \mathbb{P}_{K}
\\  &\left[
(1-k\mathbb{P}_{K} 
)^{m-1} - \sum^{K}_{\kappa=k+1} (m-1)\mathbb{P}_{K}  (1- \kappa\mathbb{P}_{K} 
)^{m-2} 
\right].
\end{align}

Define $\tau(m)=\sum^{K-1}_{k=1} \frac{M-1}{M} m     
  \mathbb{P}_{K}
[
(1-k\mathbb{P}_{K} 
)^{m-1} - \sum^{K}_{\kappa=k+1} (m-1)\mathbb{P}_{K}  (1- \kappa\mathbb{P}_{K} 
)^{m-2} 
]$. Note that in the expression of $\tau(m)$, $k\geq 1$, and hence $\kappa \geq 2$. For the case of $K=2$, $1- 2\mathbb{P}_{K}=0$. By using this observation, $\tau(m)$ can  be simplified as follows:
\begin{align}
 \tau(m)= 
\frac{m(M-1)}{M}       \mathbb{P}_{K}
(1-\mathbb{P}_{K} 
)^{m-1}   ,
\end{align}  
for $m>2$, and  
\begin{align}
\tau(2)=  \left(
\frac{m}{M} (m-1)    +\frac{M-m}{M} m     
\right) \mathbb{P}_{K}
 \left[
(1-\mathbb{P}_{K} 
)  -  \mathbb{P}_{K}  
\right]=0.
\end{align}
By using the simplified expression of $\tau(m)$, $P_{0,1}$ can be simplified as follows: 
\begin{align}
P_{0,1} \approx &  M\mathbb{P}_{\rm TX}  (1- \mathbb{P}_{\rm TX} )^{M-1} \frac{M-1}{M}K\mathbb{P}_{K}
\\\nonumber &+\sum^{M}_{m=3}{M \choose m} \mathbb{P}_{\rm TX} ^{m} (1- \mathbb{P}_{\rm TX} )^{M-m}  \frac{m(M-1)}{M}       \mathbb{P}_{K}
(1-\mathbb{P}_{K} 
)^{m-1}   \\\nonumber 
=&(M-1)  \mathbb{P}_{\rm TX}  (1- \mathbb{P}_{\rm TX} )^{M-1}    +(M-1) \sum^{M}_{m=3}{M-1 \choose m-1}   \mathbb{P}_{\rm TX} ^{m} (1- \mathbb{P}_{\rm TX} )^{M-m}       \mathbb{P}_{K}^m,
\end{align} 
where the last step follows by the fact that $1-\mathbb{P}_{K}=\mathbb{P}_{K}$.

Finally, $P_{0,2} $ can rewritten as follows: 
\begin{align}
&P_{0,2}   
  \approx{M \choose 2} \mathbb{P}_{\rm TX} ^{i} (1- \mathbb{P}_{\rm TX} )^{M-2}    \sum^{K-1}_{k_1=1}\sum^{K}_{k_2=k_1+1} \frac{M-i}{M}   \mathbb{P}_{K}^2 {k_2-k_1-1 \choose 0}\prod^{1}_{p=0}   (2-p)
 \\\nonumber &+\sum^{M}_{m=3}{M \choose m} \mathbb{P}_{\rm TX} ^{m} (1- \mathbb{P}_{\rm TX} )^{M-m}    \sum^{K-2}_{k_1=1}\sum^{K-1}_{k_2=k_1+1} \frac{M-i}{M}   \mathbb{P}_{K}^2 {k_2-k_1-1 \choose 0}\prod^{1}_{p=0}   (m-p)
\\\nonumber &\times \left[
(1-k_2\mathbb{P}_{K} 
)^{m-2} - \sum^{K}_{\kappa=k_2+1} (m-2)\mathbb{P}_{K}  (1- \kappa\mathbb{P}_{K} 
)^{m-3} 
\right]  
\\\nonumber
&  \overset{(a)}{=}{M \choose 2} \mathbb{P}_{\rm TX} ^{2} (1- \mathbb{P}_{\rm TX} )^{M-2}   
 2\frac{M-2}{M}   \mathbb{P}_{K}^2  
  =\frac{(M-1)!}{(M-3)!} \mathbb{P}_{\rm TX} ^{2} (1- \mathbb{P}_{\rm TX} )^{M-2}   
   \mathbb{P}_{K}^2 ,
\end{align} 
where   step $(a)$ follows by employing  the properties of the binomial coefficients.

\subsection{Asymptotic Studies of AoI}
By using the above transition probabilities,   
the probability for the event that ${\rm U}_1$ cannot complete an update within one frame, $P_{\rm fail}$, can be simplified as follows:
 \begin{align}
P_{\rm fail} \approx& \mathbb{P}\left( Z>N\right)  = \mathbf{s}_0^T\mathbf{P}_M^N\mathbf{1} =\sum^{2}_{j=0}P_{0,j}
\\\nonumber =&  1-M  (x+1- \mathbb{P}_{\rm TX} )^{M-1}x-
 M   (1- \mathbb{P}_{\rm TX} )^{M-1}   x +2
 (M-1)   (1- \mathbb{P}_{\rm TX} )^{M-1}x   \\\nonumber &+\underset{\tau_0}{\underbrace{(M-1) \sum^{M}_{i=2}{M-1 \choose i}    (1- \mathbb{P}_{\rm TX} )^{M-i-1}     x^{i+1}}}
 +\frac{(M-1)!}{(M-3)!}  (1- \mathbb{P}_{\rm TX} )^{M-2}   
  x^2 ,
  \end{align}
where $x=\mathbb{P}_{\rm TX}\mathbb{P}_{K}$. By using the properties of binomial coefficients, $\tau_0$ can be rewritten as follows:
\begin{align}
\tau_0 =&x(M-1) \sum^{M}_{i=2}{M-1 \choose i}    (1- \mathbb{P}_{\rm TX} )^{M-i-1}     x^{i}
\\\nonumber =& 
(1-\mathbb{P}_{\rm TX}+x)^{M-1}(M-1)x- (1-\mathbb{P}_{\rm TX})^{M-1}(M-1)x\\\nonumber &-(M-1)(1-\mathbb{P}_{\rm TX})^{M-2}(M-1)x^2.
\end{align}
By using the simplified expression of $\tau_0$, $P_{\rm fail}$ can be rewritten as follows: 
 \begin{align}\nonumber
P_{\rm fail} \approx&  1-M  (x+1- \mathbb{P}_{\rm TX} )^{M-1}x-
 M   (1- \mathbb{P}_{\rm TX} )^{M-1}   x +2
 (M-1)   (1- \mathbb{P}_{\rm TX} )^{M-1}x   \\\nonumber &+(1-\mathbb{P}_{\rm TX}+x)^{M-1}(M-1)x- (1-\mathbb{P}_{\rm TX})^{M-1}(M-1)x\\\nonumber &-(M-1)(1-\mathbb{P}_{\rm TX})^{M-2}(M-1)x^2
 +(M-1)(M-2)  (1- \mathbb{P}_{\rm TX} )^{M-2}   
  x^2   \\\nonumber
  =&  1-  (\mathbb{P}_{\rm TX}\mathbb{P}_{K}+1- \mathbb{P}_{\rm TX} )^{M-1}\mathbb{P}_{\rm TX}\mathbb{P}_{K} -
  ( 1- \mathbb{P}_{\rm TX}   +(M-1)\mathbb{P}_{\rm TX}\mathbb{P}_{K}) (1-\mathbb{P}_{\rm TX})^{M-2} \mathbb{P}_{\rm TX}\mathbb{P}_{K} .
  \end{align}


The simplified expression of $P_{\rm fail}$ can be used to facilitate the asymptotic studies of the AoI. 
For the case of $N=1$,    the pmf of the  access delay can be obtained from $P_{\rm fail}$ as follows:
\begin{align}
\mathbb{P}(S_j=T) =  \frac{\mathbf{s}_0^T\mathbf{P}_M^{0} \mathbf{p}}{1- \mathbf{s}_0^T\mathbf{P}_M\mathbf{1} }= \frac{p_{0,M}}{p_{0,M}}=1,
\end{align}
where $p_{0,M}$ is the first element of $\mathbf{p}$. As a result, $\mathcal{E}\{S_j\} $ and $\mathcal{E}\{S_j^2\} $ can be obtained as follows: 
\begin{align}
\mathcal{E}\{S_j\} = T\sum^{1}_{n=1}n  \frac{\mathbf{s}_0^T\mathbf{P}_M^{n-1} \mathbf{p}}{1- \mathbf{s}_0^T\mathbf{P}_M^N\mathbf{1} }=T,
\end{align}
and
\begin{align}
\mathcal{E}\{S_j^2\} = T^2\sum^{1}_{n=1}n^2  \frac{\mathbf{s}_0^T\mathbf{P}_M^{n-1} \mathbf{p}}{1- \mathbf{s}_0^T\mathbf{P}_M^N\mathbf{1} }=T^2,
\end{align}
 which are expected since $S_j$ becomes a deterministic parameter for this special case.

Furthermore, recall that the inter-departure time $Y_j$ can be rewritten as follows:
\begin{align}
Y_j =  \left(NT-S_{j-1}\right)+ (X_j-1)NT + S_{j}.
\end{align}
Therefore, for the case of $N=1$,   the expectation of $Y_j$ can be simplified as follows: 
\begin{align}
\mathcal{E}\{Y_j\} =& \mathcal{E}\{ \left(NT-S_{j-1}\right)+ (X_j-1)NT + S_{j}\}\\\nonumber
=&
T N   \mathcal{E}\{X_{j}\} = TN\frac{1}{1-P_{\rm fail}} ,
\end{align}
which is obtained by using the    fact that $X_j$ follows the geometric  distribution. 
Similarly, the expectation of $Y_j^2$ can be simplified as follows: 
\begin{align}
\mathcal{E}\{Y_j^2\}  
=&
   N^2T^2 \mathcal{E}\left\{X_j^2 \right\} 
    + 2\mathcal{E}\left\{S_{j}^2 \right\}    - 2  \mathcal{E}\left\{S_{j}\right\}^2   =
   N^2T^2 \frac{1+P_{\rm fail}}{\left(1-P_{\rm fail}\right)^2}
    + 2T^2   - 2T ^2   .
\end{align}
 
Finally, for the case of $N=1$,   the averaged AoI can be expressed as follows: 
\begin{align}
\bar{\Delta}^N =&  \frac{\mathcal{E}\{S_{j}\}\mathcal{E}\{Y_j\} -\mathcal{E}\{S_{j} ^2 \}+\mathcal{E}\{S_{j}\}  ^2}{\mathcal{E}\{Y_j\}}+\frac{\mathcal{E}\{Y_j^2\}}{2\mathcal{E}\{Y_j\}}=  \frac{T  \mathcal{E}\{Y_j\} -T^2+T^2}{\mathcal{E}\{Y_j\}}+\frac{\mathcal{E}\{Y_j^2\}}{2\mathcal{E}\{Y_j\}}
\\\nonumber 
=& T+\frac{   N^2T^2 \frac{1+P_{\rm fail}}{\left(1-P_{\rm fail}\right)^2}
  }{2TN\frac{1}{1-P_{\rm fail}}}\approx T+\frac{   NT (2-f(\mathbb{P}_{\rm TX}))}{2f(\mathbb{P}_{\rm TX})},
\end{align}
where  $f(\mathbb{P}_{\rm TX})= (\mathbb{P}_{\rm TX}\mathbb{P}_{K}+1- \mathbb{P}_{\rm TX} )^{M-1}\mathbb{P}_{\rm TX}\mathbb{P}_{K}+
  ( 1- \mathbb{P}_{\rm TX}   +(M-1)\mathbb{P}_{\rm TX}\mathbb{P}_{K}) (1-\mathbb{P}_{\rm TX})^{M-2} \mathbb{P}_{\rm TX}\mathbb{P}_{K} $. 

 \subsection{Finding the Optimal Choice of $\mathbb{P}_{\rm TX}$}
The considered AoI minimization problem can be expressed as follows: $\underset{0\leq \mathbb{P}_{\rm TX}\leq 1}{\min }\bar{\Delta}^N$.  It is challenging to show whether   $\bar{\Delta}^N$ is     a convex function of $\mathbb{P}_{\rm TX}$, given the complex expression of $\bar{\Delta}^N$. However, it can be shown that $\bar{\Delta}^N$ first decreases and then increases as $\mathbb{P}_{\rm TX}$ grows, as shown in the following. 

The first order derivative of $\bar{\Delta}^N$ with respect to $\mathbb{P}_{\rm TX}$ is given by
\begin{align}\nonumber
\left(\bar{\Delta}^N\right)' \approx& \frac{ NT}{2} \left( \frac{   -f'(\mathbb{P}_{\rm TX})}{f(\mathbb{P}_{\rm TX})}-\frac{   2f'(\mathbb{P}_{\rm TX})'-f(\mathbb{P}_{\rm TX})f'(\mathbb{P}_{\rm TX})}{f^2(\mathbb{P}_{\rm TX})}\right)=- \frac{ NT}{2} \frac{   2f'(\mathbb{P}_{\rm TX})}{f^2(\mathbb{P}_{\rm TX})} ,
\end{align}
 which shows that the monotonicity of $\bar{\Delta}^N$ is decided by the sign of $f'(\mathbb{P}_{\rm TX})$ only.   
In the following, we will first show that $f'(\mathbb{P}_{\rm TX})=0$ has  a single root    for $0\leq \mathbb{P}_{\rm TX}\leq 1$. 

With some straightforward algebraic manipulations, $f'(\mathbb{P}_{\rm TX})$ can be expressed as follows:
{\small \begin{align}
f'(\mathbb{P}_{\rm TX})
  =&\left[1-M\frac{\mathbb{P}_{\rm TX}}{2}\right]\left( 1-\frac{\mathbb{P}_{\rm TX}}{2}\right)^{M-2} +
\left[ 1-2\mathbb{P}_{\rm TX}-M(M-3)\frac{\mathbb{P}_{\rm TX}^2}{2} \right]\left(1-\mathbb{P}_{\rm TX}\right)^{M-3}.
\end{align}}
Because $0\leq \mathbb{P}_{\rm TX}\leq 1$, both $\left( 1-\frac{\mathbb{P}_{\rm TX}}{2}\right)$ and  $\left(1-\mathbb{P}_{\rm TX}\right)$ are positive. 

On the one hand, if $\mathbb{P}_{\rm TX}\geq\frac{2}{M}$, $1-M\frac{\mathbb{P}_{\rm TX}}{2}\leq 0$, otherwise $1-M\frac{\mathbb{P}_{\rm TX}}{2}\geq 0$. On the other hand, there are two roots for $1-2\mathbb{P}_{\rm TX}-M(M-3)\frac{\mathbb{P}_{\rm TX}^2}{2} =0$, namely  $\frac{\sqrt{4+2M(M-1)}-4}{M(M-3)}$ and $\frac{-\sqrt{4+2M(M-1)}-4}{M(M-3)}$. When $M\rightarrow \infty$, the two roots can be approximated as follows: 
\begin{align}
  \frac{\pm\sqrt{4+2M(M-1)}-4}{M(M-3)} &\approx \pm\frac{\sqrt{2}}{M}.
\end{align}
 Therefore,   $1-2\mathbb{P}_{\rm TX}-M(M-3)\frac{\mathbb{P}_{\rm TX}^2}{2} \leq 0$ if $\mathbb{P}_{\rm TX}\geq \frac{\sqrt{2}}{M}$, otherwise $1-2\mathbb{P}_{\rm TX}-M(M-3)\frac{\mathbb{P}_{\rm TX}^2}{2} > 0$. Denote the root(s) of  $f'(\mathbb{P}_{\rm TX})=0$ by $\mathbb{P}_{\rm TX}^*$, which  is bounded as follows:
  \begin{align}\label{boundx}
 \frac{\sqrt{2}}{ M} \leq \mathbb{P}_{\rm TX}^* \leq \frac{2}{M}. 
\end{align}

A key observation from  \eqref{boundx} is that the upper and lower bounds on $\mathbb{P}_{\rm TX}^*$ are of the same order of $\frac{1}{M}$. Therefore,     $\mathbb{P}_{\rm TX}^*$  can be expressed as $\mathbb{P}_{\rm TX}^*=\frac{\eta}{M}$, where $\sqrt{2}\leq \eta\leq 2$. By using this expression for  $\mathbb{P}_{\rm TX}^*$,      $f'(\mathbb{P}^*_{\rm TX})=0$ can be expressed as follows: 
\begin{align}\label{froot}
0=&\left(1-\frac{\eta}{2}\right)\left( 1-\frac{\eta}{2M}\right)^{M-2} +
\left(  1-2\frac{\eta}{M}-(M-3)\frac{\eta^2}{2M} \right)\left(1-\frac{\eta}{M}\right)^{M-3}.
\end{align}  

In order to find an explicit expression of $\mathbb{P}_{\rm TX}^* $, we note that  $f_x(x)=\left(1-\frac{a}{x}\right)^x$ can be approximated at $x\rightarrow \infty$ as follows:
\begin{align}
\ln f_x(x) =& \underset{x\rightarrow \infty}{\lim} x\ln \left(1-\frac{a}{x}\right)=\underset{z\rightarrow 0}{\lim} \frac{\ln \left(1-az\right)}{z} =-\underset{z\rightarrow 0}{\lim}\frac{a}{1-az}=-a, 
\end{align}
which implies  $\left(1- \frac{a}{M} \right)^{M}=e^{-a}$ for large $M$. 
  
Therefore, for $M\rightarrow \infty$, \eqref{froot} can be approximated as follows: 
   \begin{align}\label{etaeq}
0=&\left(1-\frac{\eta}{2}\right)e^{-\frac{\eta}{2}}  +\left(1-\frac{\eta^2}{2}\right) e^{-\eta},
 \end{align}  
where the value of $\eta$  can be straightforwardly obtained by applying   off-the-shelf root solvers. It is important to point out that \eqref{etaeq} has  a single root, which means that   $\mathbb{P}_{\rm TX}^*=\frac{\eta}{M}$ is the single root of $f'(\mathbb{P}_{\rm TX})=0$. Therefore,  $\mathbb{P}_{\rm TX}^*=\frac{\eta}{M}$ is the optimal choice of $\mathbb{P}_{\rm TX}$ to minimize  $\bar{\Delta}^N$, and  the proof is complete.

\section{Proof for Proposition \ref{proposition}}\label{proof3}
Based on Lemma \ref{lemma2}, the optimal choice for the transmission probability is given by $\mathbb{P}_{\rm TX}^*=\frac{\eta}{M}$. By using  this optimal choice of $\mathbb{P}_{\rm TX}$,   
$P_{\rm fail}$ can be expressed as follows:
\begin{align}
P_{\rm fail} \approx&    
  1-  \left( 1- \frac{\eta}{2M}\right )^{M-1}\frac{\eta}{2M} -
  \left( 1-\frac{\eta}{M}  +(M-1)\frac{\eta}{2M}\right) \left(1-\frac{\eta}{M}\right)^{M-2} \frac{\eta}{2M}.
  \end{align}
  
  Again applying the following limit:   $\underset{M\rightarrow \infty}{\lim}\left(1- \frac{a}{M} \right)^{M}=e^{-a}$, $P_{\rm fail}$ can be approximated for large $M$ as follows:
{\small \begin{align} 
P_{\rm fail} 
  \approx &   
  1-  e^{-\frac{\eta}{2}}\frac{\eta}{2M} -
  \left( 1-\frac{\eta}{M}  +(M-1)\frac{\eta}{2M}\right) e^{-\eta} \frac{\eta}{2M}
  \approx     
  1-  e^{-\frac{\eta}{2}}\frac{\eta}{2M} -
  \left( 1+\frac{\eta}{2}\right)  e^{-\eta} \frac{\eta}{2M}. 
  \end{align}}
By using this approximation of $P_{\rm fail} $, the AoI achieved by NOMA can be approximated as follows: 
 \begin{align}
\bar{\Delta}^N  
\approx&    T+\frac{NT}{2}\frac{ 2-  e^{-\frac{\eta}{2}}\frac{\eta}{2M} -
  \left( 1+\frac{\eta}{2}\right)  e^{-\eta} \frac{\eta}{2M}}{  e^{-\frac{\eta}{2}}\frac{\eta}{2M} +
  \left( 1+\frac{\eta}{2}\right)  e^{-\eta} \frac{\eta}{2M} }
  \\\nonumber 
\approx&    \eta+\frac{NT}{2}\left(\frac{ 2 }{  e^{-\frac{\eta}{2}}\frac{\eta}{2M} +
  \left( 1+\frac{\eta}{2}\right)  e^{-\eta} \frac{\eta}{2M} }-1\right)
\approx  NT\frac{ 2Me^{\eta} }{ \eta\left( e^{\frac{\eta}{2}} +
   1+\frac{\eta}{2} \right)    } .
\end{align}

In order to find an explicit expression for the AoI achieved by  OMA  for the case of $N=1$, the transition probabilities can be expressed as follows: 
\begin{align}
P_{0,0} = 1- M\mathbb{P}_{\rm TX} \left( 1-\mathbb{P}_{\rm TX}\right)^{M-1},
\end{align}
and 
\begin{align}
P_{0,1} = (M-1)\mathbb{P}_{\rm TX} \left( 1-\mathbb{P}_{\rm TX}\right)^{M-1}. 
\end{align}
Therefore, for OMA, the failure probability,  $P_{\rm fail}$, is given by
\begin{align}
P_{\rm fail} =& P_{0,0}+P_{0,1}=
1- \mathbb{P}_{\rm TX} \left( 1-\mathbb{P}_{\rm TX}\right)^{M-1}.
\end{align}
By using the expression for $P_{\rm fail}$, the AoI achieved by OMA can be expressed as follows:
\begin{align}
\bar{\Delta}^O = &1+\frac{   NT (1+P_{\rm fail})}{2(1-P_{\rm fail})}
=1+\frac{   NT \left(2- \mathbb{P}_{\rm TX} \left( 1-\mathbb{P}_{\rm TX}\right)^{M-1}\right)}{2 \mathbb{P}_{\rm TX} \left( 1-\mathbb{P}_{\rm TX}\right)^{M-1}}.
\end{align}
Define $f_O( \mathbb{P}_{\rm TX} )=\mathbb{P}_{\rm TX} \left( 1-\mathbb{P}_{\rm TX}\right)^{M-1}$. Similar to the NOMA case, it is straightforward to show that the optimal choice of $\mathbb{P}_{\rm TX}$ for OMA is the root of $f'_O( \mathbb{P}_{\rm TX} )=0$. Note that $f'_O( \mathbb{P}_{\rm TX} )$ can be expressed as follows: \begin{align}
 f'_O( \mathbb{P}_{\rm TX} )=& (1-M\mathbb{P}_{\rm TX} )(1-\mathbb{P}_{\rm TX} )^{M-2},
 \end{align}
which is the reason why $\mathbb{P}_{\rm TX} ^*=\frac{1}{M}$ is optimal for the OMA case. By using $\mathbb{P}_{\rm TX} ^*=\frac{1}{M}$, $P_{\rm fail}$ can be expressed as follows: $
P_{\rm fail} =
1- \frac{1}{M}\left( 1-\frac{1}{M}\right)^{M-1}$,
and hence the AoI achieved by OMA is given by
 \begin{align}
\bar{\Delta}^O = &1+\frac{NT}{2}\frac{  2- \mathbb{P}_{\rm TX} \left( 1-\mathbb{P}_{\rm TX}\right)^{M-1}}{ \mathbb{P}_{\rm TX} \left( 1-\mathbb{P}_{\rm TX}\right)^{M-1}}
= 1+\frac{NT}{2}\left[\frac{  2M}{  \left( 1-\frac{1}{M}\right)^{M-1}}-1\right].
\end{align}
For $M\rightarrow \infty$, the AoI achieved by OMA can be approximated as follows:
 \begin{align}
\bar{\Delta}^O \approx  &  \frac{NT}{2} \frac{  2M}{  \left( 1-\frac{1}{M}\right)^{M-1}} 
\approx   NTMe.
\end{align}

Therefore, for $M\rightarrow \infty$, the ratio between the AoI realized by NOMA and OMA is given by
\begin{align}
 \frac{\bar{\Delta}^N }{\bar{\Delta}^O }\approx \frac{NT\frac{ 2Me^{\eta} }{ \eta\left( e^{\frac{\eta}{2}} +
   1+\frac{\eta}{2} \right)    } }{NTMe}= \frac{ 2e^{\eta-1} }{ \eta\left( e^{\frac{\eta}{2}} +
   1+\frac{\eta}{2} \right)    },
\end{align}
which completes the proof. \vspace{-1em}

 \bibliographystyle{IEEEtran}
\bibliography{IEEEfull,trasfer}
  \end{document}